\documentclass[11pt]{article}

\usepackage[margin=1in]{geometry}
\usepackage{amsmath}
\usepackage{amsthm}
\usepackage{amssymb}
\newtheorem{theorem}{Theorem}[section]
\newtheorem{lemma}{Lemma}[section]

\newtheorem{proposition}{Proposition}[section]
\newtheorem{observation}{Observation}[section]
\theoremstyle{definition}
\newtheorem{definition}{Definition}[section]
\newtheorem{remark}{Remark}[section]
\newtheorem{example}{Example}[section]
\usepackage{tabularx}
\usepackage{xcolor}
\usepackage{hyperref}
\hypersetup{
    colorlinks,
    linkcolor={red!80!black},
    citecolor={blue!80!black},
    urlcolor={blue!80!black}
}
\usepackage{enumitem}

\usepackage{natbib}
\usepackage{booktabs} 
\usepackage[ruled]{algorithm2e}
\SetAlFnt{\small}
\SetAlCapFnt{\small}
\SetAlCapNameFnt{\small}
\SetAlCapHSkip{0pt}
\IncMargin{-\parindent}
\usepackage{subcaption}

\newcommand{\contract}{\langle t,i \rangle}
\newcommand{\ambcontract}{\langle \tau, i \rangle}
\newcommand{\ambcontractfull}{\langle \{t^1, \ldots, t^k\},i \rangle}

\newcommand{\ambcontractprime}{\langle \tau', i \rangle}
\newcommand{\UA}[1]{U_A(#1)}
\newcommand{\UAmid}[2]{U_A({#1} \mid {#2})}
\newcommand{\UP}[1]{U_P(#1)}
\newcommand{\UPmid}[2]{U_P({#1} \mid {#2})}
\newcommand{\Reward}{R}
\newcommand{\Welfare}{W}
\newcommand{\Payment}[2]{T_{#1}(#2)}

\DeclareMathOperator*{\argmax}{arg\,max}
\DeclareMathOperator*{\argmin}{arg\,min}

\newcommand{\Wel}[0]{W\xspace}

\DeclareRobustCommand{\stirling}{\genfrac\{\}{0pt}{}}

\newcommand{\nminusi}{[n] \setminus \{i\}}
\newcommand{\tauactioni}[0]{\langle \tau, i \rangle}
\newcommand{\tausizek}[0]{\tau = \{t^1, \ldots, t^k\}}

\newcommand{\tauhatsizek}[0]{\hat{\tau} = \{\hat{t}^1, \ldots, \hat{t}^k\}}

\newcommand{\Matching}{\textsc{3D-Matching}\xspace}

\usepackage[]{color-edits}

\addauthor{YR}{red}

\addauthor{MF}{magenta}

\addauthor{PD}{blue}

\makeatletter
\AddToHook{cmd/appendix/before}{\def\cref@section@alias{appendix}\def\cref@subsection@alias{appendix}}

\title{Succinct Ambiguous Contracts\thanks{We would like to thank Gil Cohen, Daniel Peretz, Larry Samuelson, and Tami Tamir for insightful discussions and pointers. This project has been partially funded by the European Research Council (ERC) under the European Union's Horizon 2020 research and innovation program (grant agreement No. 866132), and under the 
European Union's Horizon Europe Program (grant agreement No.~101170373), by an Amazon Research Award, by the NSF-BSF (grant number 2020788), by the Israel Science Foundation Breakthrough Program (grant No.2600/24), and by a grant from TAU Center for AI and Data Science (TAD).}}

\author{Paul D\"utting\thanks{Google Research, Zurich, Switzerland. Email: \texttt{duetting@google.com}} \and Michal Feldman\thanks{Blavatnik School of Computer Science, Tel Aviv University, Israel. Email: \texttt{mfeldman@tauex.tau.ac.il}} \and Yarden Rashti\thanks{Blavatnik School of Computer Science, Tel Aviv University, Israel. Email: \texttt{rashtiyarden@mail.tau.ac.il}}}
\date{}

\begin{document}

\maketitle

\begin{abstract}
Real-world contracts are often ambiguous. While recent work by D\"utting, Feldman, Peretz, and Samuelson (EC 2023, Econometrica 2024) demonstrates that ambiguous contracts can yield large gains for the principal, their optimal solutions often require deploying an impractically large menu of contracts. This paper investigates \emph{succinct} ambiguous contracts, which are restricted to consist of at most $k$ classic contracts. By letting $k$ range from $1$ to $n-1$, this yields an interpolation between classic contracts ($k=1$) and unrestricted ambiguous contracts ($k=n-1$).
This perspective enables important structural and algorithmic results.
First, we establish a fundamental separability property: for any number of actions $n$ and any succinctness level $k$, computing an optimal $k$-ambiguous contract reduces to finding optimal classic contracts over a suitable partition of the actions, up to an additive balancing shift acting as a base payment. Second, we show bounds on the principal's loss from using a $k$-ambiguous rather than an unrestricted ambiguous contract, which uncover a striking discontinuity in the principal's utility regarding contract size: lacking even a single contract option may cause the principal's utility to drop sharply by a multiplicative factor of $2$, a bound we prove to be tight. Finally, we characterize the tractability frontier of the optimal $k$-ambiguous contract problem.
Our separability result yields a poly-time algorithm whenever the number of partitions of $n-1$ actions into $k$ sets is polynomial, recovering and extending known results for classic and unrestricted ambiguous contracts.
We complement this with a tight hardness result, showing that the problem is \textsf{NP}-hard whenever the number of partitions is super-polynomial. Moreover, when $k \approx n/3$, the problem is even hard to approximate.
\end{abstract}

\section{Introduction}
\label{sec:intro}
Contract design is a cornerstone of microeconomic theory \cite[][]{Holmstrom79,Ross73,Mirrlees99,GrossmanHart83}, as evidenced by the 2016 Nobel Prize in Economics awarded to Hart and Holmstr\"om \citep{NobelPrize16}. Just as mechanism design is fundamental to markets for goods, contract design serves a crucial role in markets for services. In recent years, it has also gained growing interest within the computer science community, see the recent surveys of \cite{DuttingFT24} and \cite{Feldman26}.

Real-world contracts are frequently ambiguous, as reflected in the inherently ambiguous terms they often contain. A software development contract might stipulate that the code be written in a ``maintainable and efficient" manner, an employment contract might state that an employee must demonstrate ``professionalism and commitment," and a university tenure policy may require a faculty member to show ``impactful contributions in their field." In each of these cases, the precise meaning of the terms is open to interpretations. 

The potential role of ambiguity in contracts was first noted by \citet{BernheimW98}. 
More recently, \citet{DuettingFPS24} proposed a model of ambiguous contracts, inspired by a similar framework for ambiguous auctions \citep{DiTillioEtAl17}.

Their model builds on a  
standard principal-agent framework with $n$ actions and $m$ outcomes. Each action $i$  
has an associated cost $c_i$, and an outcome distribution represented by a vector $p_i$, where $p_{ij}$ denotes the probability of outcome $j$ under action $i$. The principal observes only the realized outcomes, and, to incentivize the agent, designs a contract --- a payment function $t:[m]\rightarrow \mathbb{R}_+$ that specifies a payment for each possible outcome. The agent's utility for action $i$ is given by the expected payment for action $i$ minus the action's cost $c_i$.

The key innovation in their model is the introduction of ambiguity. An ambiguous contract $\tau=\{t^1, t^2, \ldots\}$ consists of a collection of payment functions, with the agent evaluating each action based on the worst-case utility over the contract’s support. This max-min evaluation captures the agent’s ambiguity aversion \citep{Schmeidler89,GilboaS93}. 
To ensure the credibility of the principal's offer, an ambiguous contract must satisfy a \emph{consistency} condition, which requires that the principal be indifferent among all contracts within the ambiguous contract, given the agent's best response. 
By leveraging this form of ambiguity, the principal can typically incentivize actions at a lower payment, or implement actions that would otherwise be infeasible to implement with  
a classic contract. 

As demonstrated by \citet{DuettingFPS24}, 
optimal ambiguous contracts exhibit structural simplicity: 
they are composed of 
at most $n-1$ 
single-outcome payment (SOP) functions --- payment functions that assign payment to only a single outcome. However, while each SOP contract is inherently simple, a contract composed of many such functions can become increasingly complex.
Under the Monotone Likelihood Ratio Property (MLRP) condition,
two classic contracts are sufficient to fully exploit the benefits of ambiguity. 
However, MLRP is a very strong assumption, much stronger than First-Order Stochastic Dominance (FOSD), and many practical scenarios 
violate it.
In contrast, in the general case, achieving the full potential of ambiguous contracts may require (roughly) as many contracts as there are actions or outcomes.\footnote{The MLRP condition \cite[see, e.g.,][Definition 2]{DuettingFPS24} requires that higher output must always be more likely to have come from high effort than low effort. In contrast, FOSD only requires that higher effort shifts the distribution to better outcomes.}  

\paragraph{Succinct Ambiguous Contracts.} 
A natural way to manage this growing complexity is to focus on ambiguous contracts with small support. Specifically, a 
$k$-ambiguous contract is an ambiguous contract consisting of at most 
$k$-classic contracts. Indeed, while it may be unrealistic for a contract to offer a vast number of different compensation schemes, it is quite common for contracts to specify a limited number of possibilities.

For example, sales commission contracts often define multiple bonus structures: one for meeting a baseline sales target and another, higher bonus, for exceeding a stretch goal. However, the exact criteria for receiving the higher bonus may not be fully transparent, as it may depend on subjective performance evaluations or discretionary factors determined by management.

It is not difficult to see that when ambiguous contracts are restricted to a limited size, the optimal contract is no longer composed solely of SOP classic contracts (see, e.g., Example~\ref{ex:separation}). This raises interesting structural and computational questions: What is the structure of optimal $k$-ambiguous contracts? 
What's the impact of using ambiguous contracts of different sizes? Can we efficiently compute optimal or near-optimal $k$-ambiguous contracts?

\subsection{Our Results}

\paragraph{Characterization and Separability.}

Our first result gives a structural characterization of optimal $k$-ambiguous contracts for any $k \ge 1$.
Its key structural contribution is a surprising \emph{separability} property: up to a balancing operation, designing an optimal $k$-ambiguous contract reduces to solving optimal classical contracts on an appropriate partition of the action set.
We emphasize that a priori there is no reason to expect such a separability property to hold.

Formally, our characterization is in terms of \emph{shifted min-pay contracts}, these are optimal classic contracts obtained by solving the standard min-pay LP (see Figure~\ref{fig:minpaylp}) for a subset of the actions, and then shifting the resulting contracts by an additive amount to ensure consistency.

\vspace{0.1in}
\noindent {\bf Theorem}
(Theorem~\ref{thm:opt-k-ambiguous}).
Fix any $k$. Suppose action $i^\star$ is implementable with a $k$-ambiguous contract. Then there is an optimal IC $k$-ambiguous contract $\langle \tau = \{t^1, \ldots, t^k\}, i^\star\rangle$ for implementing action $i^\star$ 
that takes the following form: 
\begin{itemize}[noitemsep]
\item There is a partition of the actions $[n]\setminus \{i^\star\}$ into $k$ sets $\{S^1, \ldots, S^k\}$. 
\item For each $\ell \in [k]$, contract $t^\ell$ is a shifted min-pay contract for action $i^\star$ 
for the subinstance $\mathcal{I}_{\{i^\star\} \cup S^\ell}$, obtained by restricting the original instance to actions $\{i^\star\} \cup S^\ell$.
\end{itemize}

To prove this characterization, we first define the concept of a set of $k$ payment functions that ``protect" a target action against a partition of the remaining actions into $k$ sets.
We then show how such a protective set of payment functions can be turned into an IC $k$-ambiguous contract. 
Using these observations, we show how to compute an optimal IC $k$-ambiguous contract for a given target action and a fixed partition of the remaining actions into $k$ sets. 
We conclude by showing that, by iterating over all partitions of the actions other than the target action into $k$ sets, we obtain an optimal IC $k$-ambiguous contract for the target action.

We also establish the following characterization of the {\em implementability} of an action $i$ by an IC $k$-ambiguous contract, meaning that there exists a $k$-ambiguous contract under which the agent prefers action $i$ over any other action.

\vspace{0.1in}
\noindent {\bf Proposition}
(Proposition~\ref{prop:implementable-with-k-ambiguous}). 
    An action $i$ in an instance $\mathcal{I}$ is implementable by a $k$-ambiguous contract if and only if there exists a partition $\mathcal{S}$ of $[n]\setminus\{i\}$ to $k$ sets, such that for every set $S\in \mathcal{S}$, 
    action $i$ is implementable by a classic contract in subinstance $\mathcal{I}_{\{i\} \cup S}$. 
\vspace{0.1in}

This characterization shows  
that increasing 
$k$ from $1$ to $n-1$ gradually relaxes the implementability constraints, with classic contracts and unrestricted ambiguous contracts serving as the two endpoints of this spectrum. It also shows that the procedure, which, given a target action, iterates over all partitions of the remaining actions into $k$ sets, allows us to decide whether the target action is implementable with a $k$-ambiguous contract. Namely, it is implementable if (and only if) for at least one partition all min-pay LPs are feasible.

Notably, both our characterization and implementability results generalize the characterizations established in the classic contract setting \citep{HermalinK91}, and in the setting of unrestricted ambiguous contracts \citep{DuettingFPS24} (see Remark~\ref{rem:comparison}).

\paragraph{Ambiguous Contracts of Different Sizes.}
To explore the impact of using ambiguous contracts of different sizes, we introduce the \emph{succinctness gap}.
The succinctness gap of an instance is defined as the ratio between the principal's optimal utility under $k$-ambiguous contracts and their optimal utility under unconstrained ambiguous contracts. 
The succinctness gap of a class of instances is the worst-case  
succinctness gap across all instances in the class.


We establish upper and lower bounds for this gap for all values of $k$ between $k=1$ (corresponding to classic contracts) and $k=n-1$ (corresponding to unrestricted ambiguous contracts). 
Our first result shows a general lower bound (i.e., positive result) on the succinctness gap, for all values of $k$. It shows that the principal's utility from a $k$-ambiguous contract is at least a $\frac{1}{n-k}$ fraction of the principal's utility in the unrestricted case.

\vspace{0.1in}
\noindent {\bf Theorem}
(Theorem~\ref{thm:k-lower}).
For any instance with $n$ actions, for any $k=1,\ldots,n-1$, there exists an IC $k$-ambiguous contract whose principal's utility is at least $\frac{1}{n-k}$ of the principal's utility in the unrestricted case.
\vspace{0.1in}

We also establish the following general upper bound (i.e., negative result), showing that there are instances where the optimal $k$-ambiguous achieves at most a $\frac{1}{\lfloor\frac{n-2}{k}\rfloor + 1}$ fraction of the principal's utility from an optimal unrestricted ambiguous contract.

\vspace{0.1in}
\noindent {\bf Theorem}  
(Theorem~\ref{thm:kovern-new}). 
For any $n,k$ such that $n \geq 3$ and $1 \leq k \leq n-2$ there exists an instance such that any IC $k$-ambiguous contract obtains at most $\frac{1}{\lfloor\frac{n-2}{k}\rfloor + 1}$ of the principal's utility in the unrestricted case. 

\medskip

For the case of classic contracts ($k=1$), our results thus yield a tight bound of $\frac{1}{n-1}$, recovering the ambiguity gap of \cite{DuttingFT24}.
Interestingly, our bounds are also tight for contracts with one fewer payment function than the maximum needed ($k=n-2$), where we obtain a tight bound of $\frac{1}{2}$.
Thus, in addition to recovering the tight gap between classic contracts and optimal unrestricted ambiguous contracts, 
we show 
a striking discontinuity of the principal's utility regarding ambiguous contracts of different sizes: lacking even a single contract option may cause the utility to drop sharply by a multiplicative factor of $2$.

\paragraph{Tractability and Hardness.}

On the computational side, our structural characterization leads to a natural, but potentially na\"{\i}ve, algorithm for computing optimal $k$-ambiguous contracts: this algorithm 
iterates over all actions $i$ and, for each $i$, enumerates all partitions of the remaining actions $[n]\setminus\{i\}$ into $k$ sets, solving an optimal classic contract for each part.
The running time of this algorithm is proportional to the number of partitions of a set of size $n-1$ into $k$ non-empty sets, which is given by the Stirling number of the second kind $\stirling{n-1}{k}$.
While this number is polynomial for the two extremes --- classic contracts ($k=1$) and unrestricted ambiguous contracts, where $k=n-1$ (or more generally any $k$ such that $n-k=O(1)$) ---
it is super-polynomial for any other value of $k$. 
This naturally raises the question of whether a more efficient algorithm can bypass this  combinatorial explosion.

We show that this computational barrier is inherent, implying that the na\"ive algorithm that follows from our characterization 
is essentially best possible.
We do so by establishing a
reduction from the classic $(n,k)$-\textsc{Makespan Minimization} problem \cite[e.g.,][Section~A5.2]{GareyJohnson79}, showing that computing an optimal $k$-ambiguous contract with $n$ actions is computationally equivalent to this problem. As a consequence, the optimal $k$-ambiguous contract problem is $\mathsf{NP}$-hard for any $k$ such that our algorithm's running time is super-polynomial.

\vspace{0.1in}
\noindent {\bf Theorem}
(Theorem~\ref{thm:hardness-k}). 
    The $k(n)$-ambiguous contract problem is \textsf{NP}-hard for all functions $k: [n] \rightarrow [n]$ such that $(n,k(n))$-\textsc{Makespan Minimization} is \textsf{NP}-hard.
\vspace{0.1in}

We further complete the picture of the computational landscape, by showing a hardness of approximation result for intermediate values of $k$. Namely, 
via a reduction from \textsc{3D-Matching} \cite[e.g.,][Section~3.1.2]{GareyJohnson79}, we show that when $k \approx n/3$, even approximating the optimal contract is hard.

\vspace{0.1in}
\noindent {\bf Theorem}
(Theorem~\ref{thm:hardness-approx-util}). 
There exists a $k \in \Theta(n)$ 
such that no poly-time algorithm can approximate the optimal IC $k$-ambiguous contract to within a better factor than $3/4$, 
unless $\textsf{P} = \textsf{NP}$.

\medskip

Notably, our hardness results complement the structural characterization by explaining its inherent limits.
While the characterization identifies that optimal $k$-ambiguous contracts partition the action set into groups, determining these groups explicitly 
would entail 
solving classical combinatorial optimization problems that are known to be intractable.

\paragraph{Extension to Monotone Ambiguous Contracts.}
Finally, in Appendix~\ref{app:monotone}, we show that many of our results extend to settings in which contracts are required to be \emph{monotone} (non-decreasing) in the realized reward, meaning that larger rewards must entail larger payments. We show that our main characterization of optimal $k$-ambiguous contracts extends naturally to monotone $k$-ambiguous contracts by considering monotone min-pay LPs, which augment the standard min-pay LPs with 
$m-1$ additional monotonicity constraints. Moreover, both our positive result on the succinctness gap and our hardness result for computing an optimal $k$-ambiguous contract continue to hold under this monotonicity requirement.

\subsection{Related Work}

\paragraph{Ambiguous Contracts and Ambiguity-Aversion.}
The role of strategic ambiguity as a tool in contracting was first highlighted by \citet{BernheimW98} (albeit in a rather different model and sense than in this work). They consider two-player normal-form games, in which the actions of the two players are partitioned into sets, with the interpretation that a court/some third-party can distinguish between actions from different sets but not between actions that belong to the same set. A contract in their model is a restriction of the allowed actions to a subset of the actions, with a restriction on feasible contracts arising from the partial verifiability of which actions were chosen. A complete contract narrows down the action sets to a single set in each agent's partition. They find that in static games, under pure Nash equilibria, it is without loss to consider complete contracts. However, they also show that in dynamic settings, where agents choose actions sequentially, 
optimal contracts may be incomplete.

Our model and notion of ambiguous contracts build on but differ from those in \cite{DuettingFPS24}. Like \citet{DuettingFPS24}, we augment a standard hidden-action principal-agent model with ambiguity and analyze how a principal can optimally introduce ambiguity when interacting with an ambiguity-averse agent with max-min utility. 
A main departure of our work from \citet{DuettingFPS24} 
is the explicit simplicity constraint: we limit the number of classic contracts that an ambiguous contract may be composed of. 
Our results for the size-constrained setting rely on a substantial array of new techniques and insights. In particular, all the ingredients underlying our main characterization result are novel, and they reveal a previously-hidden separability structure governing optimal ambiguous contracts of varying sizes. The same holds for our hardness results, which are the first hardness results in the context of ambiguous contracts, and establish new connections to classic combinatorial optimization problems.

The notion of ambiguous contracts in \cite{DuettingFPS24} and this work is inspired by the work of \citet{DiTillioEtAl17}, 
who introduce an analogous notion in auction design, showing that a seller can exploit ambiguity to achieve higher surplus when selling an object to an ambiguity-averse buyer. They characterize optimal ambiguous mechanisms, but also note their complexity and suggest that this may limit their applicability in practice. To address this, they exhibit a class of simple ambiguous mechanisms (which they call na\"ive ambiguous mechanisms) and show that these dominate all non-ambiguous mechanisms in terms of the expected revenue they generate. In contrast, we formulate an explicit (parametric) simplicity constraint, and characterize the \emph{optimal} ambiguous design for \emph{all} ambiguity levels.

Several other works have considered settings in which the agent holds ambiguous beliefs that the principal can potentially exploit. \citet{BSRL14} examine mechanism design problems with ambiguous communication devices in which agents have max-min preferences.  \citet{BeaucheneEtAl19} and \citet{cheng2020ambiguous}  consider Bayesian persuasion, in which the sender exploits the ambiguity aversion of the receiver. \citet{BodohCreed12} 
examine screening problems with agents who have max-min preferences.
 \citet{BoseOP06} consider auctions in which the seller and bidders may both be ambiguity-averse. \citet{lopomo2011knightian} explore moral hazard problems with
agents who have Bewley preferences. 
Our work is distinguished from these studies by examining moral hazard problems in which the agent faces
ambiguity concerning the payments attached to outcomes.

\paragraph{Simplicity and Robustness in Contract Design.}
A recurring theme in contract theory (as well as mechanism design) is the trade-off between simplicity and performance. 

\citet{Carroll15}, \citet{CarrollW22}, \citet{Kambhampati23}, and \citet{BoTang24}
examine moral hazard problems in which the principal has ambiguous beliefs about the
set of actions the agent can choose from, and establish max-min optimality of linear (a.k.a.~commission-based) contracts. \citet{DaiT22} examine a principal
who writes contracts to shape the actions of a team of agents, with the principal holding ambiguous beliefs about the actions available to the agents, again showing that linear contracts are robustly optimal.

\citet{DuttingRT19} examine moral hazard problems in which the principal has
ambiguous beliefs about the distribution of outcomes induced by the agent’s actions, once again showing max-min optimality of linear contracts. In addition, \citet{DuttingRT19} show (tight) worst-case approximation guarantees of linear contracts, relative to optimal contracts.

Our work adds to this line of work, by formulating an explicit simplicity constraint, and characterizing optimal ambiguous contracts under this constraint. However, a fundamental difference is that in our work the principal intentionally injects ambiguity into the contract, and it is the agent who is ambiguity averse.

\paragraph{Algorithmic Approaches to Contracts.} 
Another adjacent trajectory in the literature takes an
algorithmic approach to contracts. An important such direction is the literature on combinatorial contracts. This includes contracting problems with multiple agents \citep{BabaioffFN06,BNW12,DuttingEFK23,CastiglioniM023,DuttingEFK25}, many actions 
\citep{DuttingEFK21,DuttingFGT24,EzraFS24}, or with complex outcome spaces \citep{DuttingRT21}.
Another important direction considers typed contract settings, either with multi-dimensional types \citep{GuruganeshSW21,CastiglioniM021,CastiglioniM022} or single-dimensional types \citep{AlonDT21,AlonDLT23}. 
\citet{CastiglioniCLXZ25} have recently established reductions between the two settings. What we share with this line of work is the algorithmic lens on contracts, however the structural reasons and sources of complexity that arise in the design of ambiguous contracts are very different from those in combinatorial or typed contracts.

\section{Model and Preliminaries}
\label{sec:model}
We consider a hidden-action principal-agent contracting problem.
The basic ingredients of the principal-agent model apply to both classic and ambiguous contracts. For ease of exposition, we focus 
on instances, where all outcomes have non-negative rewards.\footnote{Our results in Section~\ref{sec:alg} hold for any real rewards, including negative ones. The same is true for our hardness results in Section~\ref{sec:hardness}.
However, for our results concerning gaps between ambiguous contracts of different sizes (Section~\ref{sec:gap}) non-negative rewards are essential, as was already demonstrated in \cite{DuettingFPS24}.} 

\begin{definition}[Instance] An instance $\mathcal{I} = (c,r,p)$ of the principal-agent problem with $n$ actions and $m$ outcomes is specified by:
\begin{itemize}[noitemsep]
\item A cost vector $c = (c_1, \ldots, c_n)$, where for every action $i\in [n]$, $c_i \in \mathbb{R}_+$ denotes action $i$'s cost. Costs are assumed to be sorted so that $c_1 \leq c_2 \leq \ldots \leq c_n$.
\item A reward vector $r = (r_1, \ldots, r_m)$, where for every outcome $j \in [m]$, $r_j \in \mathbb{R}_+$ denotes outcome $j$'s reward. Rewards are assumed to be sorted so that $r_1 \leq r_2 \leq \ldots \leq r_m$.
\item A probability distribution matrix 
$p=(p_1,\ldots,p_n)$, where for every action $i\in [n]$, $p_i \in \Delta^m$ denotes the probability distribution of action $i$ over outcomes, where $p_{ij}$ denotes the probability of outcome $j$ under action $i$. 
\end{itemize}
\end{definition}

The expected reward of action $i$ is given by 
$\Reward_i=\sum_{j=1}^{m}p_{ij} r_j$, and the expected welfare of action $i$ is defined as $\Welfare_i = \Reward_i - c_i$.
The agent retains the option of not participating. To capture this, we assume throughout that action $1$  
is a zero-cost action, with expected reward $\Reward_1 = 0$.

The instance $\mathcal{I}$ is known to both the principal and the agent.  The agent's action is known only by the agent, while realized outcomes are observed by both the principal and the agent.

A payment function  $t : [m]\rightarrow \mathbb R_+$ identifies a payment made by the principal to the agent upon the realization of an outcome, with the payment $t(j)$ made in response to outcome $j$ typically denoted by $t_j$.  Note that the payment is a function of the outcome rather than the actions, since the principal cannot observe the agent's action, only the outcome. As standard, we impose the \emph{limited liability} assumption, requiring that the payments are non-negative.

\subsection{Classic Contracts}\label{sec:classic-contracts}

We first describe the classic setting, in which a contract,  denoted by $\contract$, is a payment function $t$  and a recommended action $i\in [n]$.  

The interpretation is that the principal posts a contract, the agent observes the contract and chooses an action and bears the attendant cost, an outcome is drawn from the distribution over outcomes induced by that action, and the principal  makes the payment to the agent specified by the contract.   The inclusion of a recommended action in the contract allows us to capture the common presumption that the agent ``breaks ties in favor of the principal.''  

More precisely, an agent who chooses action $i'$ when facing a contract $\contract$ garners expected utility
\[
\UAmid{i'}{t}~~=~~ \sum_{j=1}^mp_{i'j} t_j-c_{i'} ~~=~~ \Payment{i'}{t}-c_{i'},
\]
given by the difference between the expected payment $\Payment{i'}{t}$ and the cost $c_{i'}$.  The resulting principal's utility is $\UPmid{i'}{t} = \Reward_{i'} - T_{i'}(t)$, defined as the difference between expected reward $\Reward_{i'}$ and  payment $\Payment{i'}{t}$.

A central definition is the following definition of an incentive compatible (or IC) contract, and the associated notion of ``implementing an action.''

\begin{definition}[IC contract] 
A contract $\contract$ is incentive compatible (IC) if 
$
i\in \argmax_{i'\in [n]} \UAmid{i'}{t},
$
in which case we say that contract $\contract$ implements action $i$.
\end{definition}

Because payments are non-negative and action $1$ has zero cost, incentive compatibility ensures that the agent secures an expected utility of at least zero, and hence implies individual rationality.

We assume the agent follows the recommendation of an incentive-compatible contract.  If the principal posts the IC contract $\contract$, the payoffs to the principal and agent are then 
\begin{align*}
  \UP{\contract}~=~\UPmid{i}{t} = \Reward_i-\Payment{i}{t} \quad \text{and} \quad 
  \UA{\contract}~=~\UAmid{i}{t} = \Payment{i}{t}-c_i.
\end{align*}

It is without loss of generality to restrict the principal to incentive compatible contracts.  The idea is that an agent facing contract $\contract$  will choose an action that maximizes her expected utility given $t$, and hence the principal might as well name such an action in the contract.

\paragraph{Optimal Classic Contracts.}

An \emph{optimal} classic contract maximizes the principal's expected utility. It's well known that optimal contracts can be computed in polynomial time by solving one linear program (LP) for each action \citep[e.g.,][]{GrossmanHart83,DuttingRT19}. The LP for action $i$ minimizes the expected payment that is required to incentivize the agent to take action $i$, and thus identifies the contract that maximizes the principal's utility from action $i$. 

The LP and its dual are given in Figure~\ref{fig:minpaylp}. We refer to these as \textsf{MIN-PAY-LP}$(i)$ and \textsf{DUAL-MIN-PAY-LP}$(i)$. The variables in the primal LP are the payments $\{t_j\}_{j \in [m]}$. The variables in the dual are $\{\lambda_{i'}\}_{i' \in [n]\setminus\{i\}}.$

\begin{figure}[ht]
\centering
\begin{subfigure}[b]{0.5\textwidth}
        \centering
        \begin{align*}
        \min \quad &\sum_{j} p_{ij} t_j \\
        \text{s.t.} \quad & \sum_j p_{ij} t_j -c_i \geq \sum_j p_{i'j} t_j -c_{i'}  &&\forall i' \neq i\\
        & t_j \geq 0 &&\forall j
    \end{align*}
    \caption{\textsf{MIN-PAY-LP}$(i)$}
    \end{subfigure}%
    ~ 
    \begin{subfigure}[b]{0.5\textwidth}
        \centering
    \begin{align*}
    \max \quad & \sum_{i'\neq i}\lambda_{i'}(c_i-c_{i'})\\  
    \text{s.t.} \quad& 
    \sum_{i'\neq i}\lambda_{i'}(p_{ij}-p_{i'j})\le p_{ij} && \forall j\\
    & \lambda_{i'}\ge 0 &&\forall i'\neq i
    \end{align*}
        \caption{\textsf{DUAL-MIN-PAY-LP}$(i)$}
    \end{subfigure}
\vspace*{-0.5cm}
\caption{The minimum payment LP for action $i$ and its dual.}
\label{fig:minpaylp}
\end{figure}

We note that \textsf{MIN-PAY-LP}$(i)$ may be feasible or infeasible. It is feasible precisely when action $i$ can be implemented by some classic contract. Whenever it is feasible, we refer to any contract $\langle t, i \rangle$ that optimally solves the LP, as a \emph{min-pay contract} for action $i$. The resulting contract is an optimal classic contract for action $i$, and by iterating over all actions we can find a globally optimal classic contract.

\begin{proposition}
[\cite{GrossmanHart83,DuttingRT19}]
    An optimal classic contract can be found by solving $n$ linear programs, namely \textsf{MIN-PAY-LP}$(i)$ for each action $i \in [n]$.
\end{proposition}

\paragraph{Implementability by Classic Contracts.}
The LP formulation can also be leveraged to yield a characterization of  actions that are implementable with classic contracts. This is cast in the following proposition.

\begin{proposition}[Implementability with classic contracts, \citet{HermalinK91}]
\label{prop:implementable}
Action $i \in [n]$ is implementable with a classic contract if and only if there does not exist a convex combination $\lambda_{i'} \in [0,1]$ of the actions $i' \neq i$ that yields the same distribution over rewards $\sum_{i' \neq i}\lambda_{i'} p_{i'j} = p_{ij}$ for all $j$ but at a strictly lower cost $\sum_{i'} \lambda_{i'} c_{i'} < c_i$.
\end{proposition}

\subsection{Ambiguous Contracts}
\label{sec:model-ambigious}

We next recall the notion of \emph{ambiguous contracts}, introduced by \cite{DuettingFPS24}.
As in the classic model, a principal commits to payments contingent on outcomes; however, she now posts a \emph{collection} of payment functions rather than a single one. 

While \cite{DuettingFPS24} allow an arbitrary number of payment functions, we consider a parameterized variant: a
\emph{$k$-ambiguous contract} contains at most $k$ payment functions.
For technical reasons, it will be convenient to define them as ambiguous contracts that contain \emph{precisely} $k$ payment functions.
This is without loss, because we treat collections of payment functions as multisets (which is, in turn, without loss). So we can always ``fill up'' a collection of payment functions that contains fewer than $k$ payment functions.\footnote{We note that, one may wish to identify the smallest $k' \leq k$ that achieves a given utility. In this case, one can run our algorithms for each $k' \leq k$ and determine the smallest $k'$ satisfying the utility requirement.}

\begin{definition}[$k$-ambiguous contract]
For $k\in \mathbb{N}$, a \emph{$k$-ambiguous} contract $\ambcontract = \ambcontractfull$ is a collection of $k$
payment functions and a recommended action.  
 If $t \in \tau$, we say $t$ is in the support of $\tau$.
\end{definition}

When the number of payment functions is unlimited, as in \cite{DuettingFPS24}, we refer to this as an \emph{unrestricted ambiguous contract}.

The agent is assumed to be ambiguity-averse, employing a max-min utility \citep{Schmeidler89,GilboaS93}. Formally, we define the agent's utility for action $i$ given a collection of payment functions $\tau$ as 
\[
U_A(i \mid \tau) = \min_{t \in \tau} U_A(i \mid t).
\]

The agent's best response to a collection of payment functions is then the action $i$ that maximizes the agent's (minimum) utility.

\begin{definition}[Best response]
Action $i$ is a \emph{best response} to a collection of payment functions $\tau$ if it holds that
\begin{align*}
U_A(i \mid \tau) \geq U_A(i' \mid \tau) \quad \text{for all $i' \neq i$.}
\end{align*}
\end{definition}

Following \cite{DuettingFPS24} we also insistent on ``consistency,'' defined as follows:

\begin{definition}[Consistency]\label{def:consistent}
A collection of payment functions $\tau$ is \emph{consistent} w.r.t.~action $i$ if
\begin{align}
\label{basic-consistency-new}
\UPmid{i}{t^{\ell}} = \UPmid{i}{t^{\ell'}} \quad\text{for all $t^\ell,t^{\ell'} \in \tau$.}
\end{align} 
\end{definition}

In other words, consistency requires that, under the given target action $i$, the principal's utility is the same for \emph{all} payment functions in the collection of payment functions $\tau$. Without this assumption, the principal's ``threat'' to use any contract in the support of the ambiguous contract is implausible.

\begin{definition}[IC $k$-ambiguous contract]\label{def:ic-ambiguous}
A $k$-ambiguous contract $\langle \tau, i \rangle$ is \emph{incentive compatible (IC)} if (i) action $i$ is a best response to $\tau$, and (ii) $\tau$ is consistent with respect to action $i$. 
In this case, we say that $\langle \tau, i \rangle$ {\em implements} action $i$.
\end{definition}

If the principal posts the IC $k$-ambiguous contract $\ambcontract = \langle\{t^1,\ldots,t^k\},i\rangle$, then, by consistency, the induced payment $T_i(\tau)$ can be defined as
\[
\Payment{i}{\tau} =\Payment{i}{t^{\ell}} = \Payment{i}{t^{\ell'}} \quad \mbox{for every} \quad \ell,\ell'
\in [k]
, 
\]
and since the agent's utility is $\UAmid{i}{t} = \Payment{i}{t} - c_i$, it holds that
\[
\UAmid{i}{t^\ell} = \UAmid{i}{t^{\ell'}} \quad \mbox{for every} \quad \ell,\ell' 
\in [k]
.
\]
Consequently, the principal's and agent's expected utilities satisfy the following for any $t \in \tau$: 
\begin{align*}
U_P(\ambcontract) = \UPmid{i}{t} ~=~ \Reward_i-\Payment{i}{\tau} \quad\text{and}\quad 
U_A(\ambcontract) =  \UAmid{i}{t} ~=~ \Payment{i}{\tau}-c_i.
\end{align*}

While insisting on consistency may appear restrictive, it is in fact without loss of generality.
\citet{DuettingFPS24} establish this for unrestricted ambiguous contracts; the following lemma
extends the result to $k$-ambiguous contracts. For completeness, we include a proof of Lemma~\ref{lem:principal-utility} in Appendix~\ref{app:model}.

\begin{lemma}[Cf.~\citet{DuettingFPS24}]
\label{lem:principal-utility}
Consider $\ambcontract$ with $|\tau| = k$ such that action $i$ is a best response to $\tau$. 
Then there exists a (consistent) incentive compatible ambiguous contract $\ambcontractprime$ with $|\tau'| = k$ from which the principal obtains expected payoff at least $\max_{t \in \tau} \UPmid{i}{t}$.
\end{lemma}

An \emph{optimal} ambiguous contract maximizes the principal’s expected utility.
\citet{DuettingFPS24}  show that optimal (unrestricted) ambiguous contracts are composed of at most $\min\{n-1,m\}$ \emph{single–outcome payment} (SOP)
functions, defined as follows.

\begin{definition}[SOP payment function]
\label{def:SOP}
A payment function $t=(t_1,\ldots,t_m)$ is a \emph{single–outcome payment function} if there exists
$j\in[m]$ with $t_j>0$ and $t_{j'}=0$ for all $j'\neq j$.
\end{definition}

\begin{theorem}[Optimal ambiguous contracts, \citet{DuettingFPS24}]\label{thm:wlogSOP}
There is an optimal ambiguous contract that is composed of at most $\min\{n-1,m\}$ SOP payment functions. 
\end{theorem}

This characterization in fact applies to each action. That is, provided that action $i$ can be implemented with an ambiguous contract, an optimal ambiguous contract for action $i$ is composed of at most $\min\{n-1,m\}$ SOP payment functions.

\cite{DuettingFPS24} also show that (unrestricted) ambiguous contracts relax the requirements for implementability relative to classic contracts:

\begin{proposition}[Implementability with ambiguous contracts, \citet{DuettingFPS24}]\label{hadyn}
An action $i$ is implementable with an ambiguous contract if and only if there is no other action $i' \neq i$ such that $p_{i'} = p_{i}$ and $c_{i'} < c_i$. 
\end{proposition}

The structure of optimal ambiguous contracts becomes more complex when  they are constrained to contain fewer than $\min\{n-1,m\}$ payment functions. The next example demonstrates this by showing a strict separation between classic contracts, $k$-ambiguous contracts, and unrestricted ambiguous contracts. The example also shows that optimal $k$-ambiguous contracts are \emph{not} composed of SOP payment functions.

\begin{figure}
    \centering
    \begin{tabular}{|c|c|c|c|c|c|c|c|c|c|}
    \hline
        \rule{0pt}{3ex} \hspace{2.0mm} rewards:  \hspace{2.0mm}&\hspace{2.0mm} $r_1 = 0$ \hspace{2.0mm}&\hspace{2.0mm} $r_2 = 0$ \hspace{2.0mm}&\hspace{2.0mm} $r_3 = 0$ \hspace{2.0mm}&\hspace{2.0mm} $r_4 = 0$ \hspace{2.0mm}&\hspace{2.0mm} $r_5 = 4$ \hspace{2.0mm}&\hspace{2.0mm} costs \hspace{2.0mm} \\[1ex] \hline
        \rule{0pt}{3ex}  action $1$: & $1/4$ & $0$ & $1/12$ & $5/12$ & $1/4$ & $c_1 = 0$  \\[1ex]
        \rule{0pt}{3ex}  action $2$: & $1/4$ & $1/12$ & $0$ & $5/12$ & $1/4$ &  $c_2 = 0$  \\[1ex] 
        \rule{0pt}{3ex}  action $3$: & $1/4$ & $1/4$ & $1/4$ & $0$ & $1/4$ & $c_3 = 0$   \\[1ex] 
        \rule{0pt}{3ex}  action $4$ & $0$ & $1/6$ & $1/6$ & $1/6$ & $1/2$ &  $c_4 = 2/3$   \\[1ex] 
          \hline
    \end{tabular}
    \caption{Separation example.}
    \label{fig:separation}
\end{figure}

\begin{example}[Separation]\label{ex:separation}
Consider the example in Figure~\ref{fig:separation}. An optimal (unrestricted) ambiguous contract is $\langle \tau, 4\rangle$ with $\tau = \{(0,4,0,0,0),(0,0,4,0,0),(0,0,0,4,0)\}$. Note that this contract is composed of three SOP payment functions, each of which is used to ``protect'' action $4$ against one of the other actions. The expected payment for action $4$ is $2/3$, and the resulting principal's utility is $2-2/3 = 4/3$. 

As we will show below, an optimal $2$-ambiguous contract in this example 
is $\langle \tau,4 \rangle$ with $\tau = \{(0,8/3,8/3,0,0),(2/9,2/9,2/9,38/9,2/9)\}$. The rationale behind this contract is as follows. The first payment function $t = (0,8/3,8/3,0,0)$ ensures that the agent prefers action $4$ over actions $1$ and $2$, while the second payment function $t = (2/9,2/9,2/9,38/9,2/9)$ makes the agent prefer action $4$ over action $3$. The expected payment for action $4$ under this $2$-ambiguous contract is $8/9$, resulting in a principal's utility of  $2 - 8/9 = 10/9 < 4/3$.

Finally, the best classic contract for implementing action $4$ is $\langle t,4\rangle$ with $t = (0,0,0,0,8/3)$, which leads to an expected payment of $4/3$ and a principal's utility of $2-4/3 = 2/3$. So with a classic contract it is better to implement any of the other actions, with $t = (0,0,0,0,0)$ for a principal's utility of $1 < 10/9$. Thus in this example, the optimal principal utility from a classic contract, $2$-ambiguous contract, and an unrestricted ambiguous contract strictly differs between the three contract classes.
\end{example}

We provide a more thorough analysis of the principal's worst-case loss in utility from using a $k$-ambiguous contract rather than an unrestricted one in Section~\ref{sec:gap}.

\section{Optimal Succinct Ambiguous Contracts}
\label{sec:alg}
Our main result in this section (Theorem~\ref{thm:opt-k-ambiguous}) gives a structural characterization of optimal $k$-ambiguous contracts for any $k \ge 1$. From this, we also derive a characterization of the actions implementable by $k$-ambiguous contracts (Proposition~\ref{prop:implementable-with-k-ambiguous}). Both characterization results hold for any (possibly negative) real rewards, and extend the corresponding characterizations for the unrestricted setting. 

Our proofs reveal a particularly interesting, and perhaps surprising, \emph{separability} property of optimal $k$-ambiguous contracts: up to a balancing step, computing an optimal $k$-ambiguous contract reduces to solving optimal classical contracts on a suitable partition of the action set.

To formally state our results, it will be useful to define the notion of a subinstance.
Fix an instance $\mathcal{I} = (c,r,p)$ of the principal-agent problem with $n$ actions and $m$ outcomes. 
Then for any subset of actions $A \subseteq [n]$ we define \emph{subinstance} $\mathcal{I}_A$ to be the original instance $\mathcal{I}$ restricted to  the set of actions $A$.
Note that a subinstance may lose the property 
of admitting 
a zero-cost action.

\subsection{Characterization and Separability}

Our characterization shows that for any succinctness level $k$ and any action $i^\star$, if action $i^\star$ can be implemented with a $k$-ambiguous contract, then there is an optimal such contract $\langle \tau, i^\star\rangle$ that takes a particularly simple form. Namely, there exists a partition of the actions other than $i^\star$ into $k$ sets $\{S^1, \ldots, S^k\}$ for $1 \leq \ell \leq k$, such that each payment function $t \in \tau$ is obtained by solving the ``min-pay LP'' in Figure~\ref{fig:minpaylp} to find an optimal classic contract for action $i^\star$ for the subinstance $\mathcal{I}_{\{i^\star\} \cup S^\ell}$, and then possibly adding a fixed amount $c_\ell \geq 0$ to each outcome $j \in [m]$. We refer to such contracts as \emph{shifted min-pay contracts}.

The purpose of the additive shifts is to ensure that all payment functions have the same expected payment for action $i^\star$, and hence lead to the same principal utility (which is required for consistency, see Definition~\ref{def:consistent}).

This show that, up to leveling, the optimal $k$-ambiguous contract problem reduces to the problem of computing classic contracts on a suitable partition of the actions.

\begin{theorem}[Optimal $k$-ambiguous contracts]\label{thm:opt-k-ambiguous}
Fix any $k$. Suppose action $i^\star$ is implementable with a $k$-ambiguous contract. Then there is an optimal IC $k$-ambiguous contract $\langle \tau = \{t^1, \ldots, t^k\}, i^\star\rangle$ for implementing action $i^\star$ 
that takes the following form: 
\begin{itemize}[noitemsep]
\item There is a partition of the actions $[n]\setminus \{i^\star\}$ into $k$ sets $\{S^1, \ldots, S^k\}$. 
\item For each $\ell \in [k]$, contract $t^\ell$ is a shifted min-pay contract for action $i^\star$ 
for the subinstance $\mathcal{I}_{\{i^\star\} \cup S^\ell}$, obtained by restricting the original instance to actions $\{i^\star\} \cup S^\ell$.
\end{itemize}
\end{theorem}

Let's illustrate Theorem~\ref{thm:opt-k-ambiguous} by applying it to Example~\ref{ex:separation}. In that example, the overall optimal $2$-ambiguous contract implements action $4$ and is obtained by considering the partition 
$\{S^1, S^2\}$ with $S^1=\{1,2\}$ and $S^2=\{3\}$. 
The optimal classic contract for implementing action $4$ in the subinstance with actions $\{4\} \cup S^1$ (i.e., actions $\{4,1,2\}$) is $(0,8/3,8/3,0,0)$, while the optimal classic contract in the subinstance with actions $\{4\} \cup S^2$ (i.e., actions $\{4,3\}$) is $(0,0,0,4,0)$. The former leads to an expected payment for action $4$ of $8/9$, while the latter has an expected payment of $6/9$. The optimal $2$-ambiguous contract is given by shifted versions of the two contracts, namely $(0,8/3,8/3,0,0)$ and $ (2/9,2/9,2/9,38/9,2/9)$, which increases the payments of the second contract by $2/9$ for all outcomes. This way the expected payment of both contracts for action $4$ is $8/9$, and the principal achieves the same expected utility of $2-8/9 = 10/9$ from either of the two contracts.

\paragraph{Proof Outline.} 
To prove Theorem \ref{thm:opt-k-ambiguous} we proceed as follows. First, in Section~\ref{sec:protecting-an-action} we define the concept of protecting a target action against a partition of the actions different from the target action into $k$ sets. Then, in Section~\ref{sec:balancing-routine}, we describe a balancing routine that combines a collection of payment functions that protect a target action against a partition of the other actions into $k$ sets into an IC $k$-ambiguous contract. Afterwards, in Section~\ref{sec:optimal-for-partition}, we show how to find an optimal IC $k$-ambiguous contract for the target action and a fixed partition of the other actions into $k$ sets. Finally, in Section~\ref{sec:optimal-overall}, we show that by iterating over all actions and partitions, we obtain an optimal ambiguous contract.

Notably, we formulate the proof of Theorem~\ref{thm:opt-k-ambiguous} in terms of algorithms mostly for clarity. We return to the computational complexity of the optimal (or approximately optimal) $k$-ambiguous contract problem in 
Section~\ref{sec:hardness}.

\begin{remark}[Comparison with Theorem~\ref{thm:wlogSOP}]\label{rem:comparison}
The attentive reader may notice an important difference between Theorem~\ref{thm:opt-k-ambiguous} and Theorem~\ref{thm:wlogSOP} from \cite{DuettingFPS24}. Without additional work, Theorem~\ref{thm:opt-k-ambiguous} does not specialize to Theorem~\ref{thm:wlogSOP} for the unrestricted case with $k \geq \min\{n-1,m\}$. This is because, as stated, our theorem does \emph{not} imply that the optimal ambiguous contract in the unrestricted setting is composed of single-outcome payment (SOP) contracts. The technical reason for this is that our balancing routine adds a fixed amount to the payments for \emph{all} outcomes. However, as we discuss in Remark~\ref{rem:adjustments} below, this difference is easy to reconcile. Through a small adjustment to the balancing routine for the special case when $k$ is not subject to any restrictions, our proof can be turned into an alternative proof of Theorem~\ref{thm:wlogSOP}.
\end{remark}

\subsection{Protecting Actions against Sets of Actions and Partitions of the Actions}\label{sec:protecting-an-action}

Our proof of Theorem~\ref{thm:opt-k-ambiguous} revolves around the idea of protecting a target action $i$ against a subset of the actions other than $i$, and more generally a partition of the actions other than $i$ into $k$ sets.
We first define the notion of protecting an action against a set of actions.

\begin{definition}\label{def:protect-against-set}
    A payment function $t$ is said to protect action $i \in [n]$ against a set of actions $S \subseteq [n] \setminus \{i\}$ if $U_A(i' \mid t) \leq U_A(i \mid t)$ for all $i' \in S$.
\end{definition}

We next define what it means for a set of $k$ payment functions to protect an action against a partition of size $k$.
Let $\mathcal{T}_k$ be the set of 
all collections of $k$ payment functions. Let $\mathcal{S}_{k,i}$ denote all partitions of $[n] \setminus \{i\}$ into $k$ sets $\{S^1, \ldots, S^k\}$. That is, for each $\ell \in [k]$, $S^\ell \subseteq [n] \setminus \{i\}.$ Moreover, $S^\ell \cap S^{\ell'} = \emptyset$ for all $\ell,\ell' \in [k]$ with $\ell \neq \ell'$ and $\bigcup_{\ell=1}^{k} S^\ell = [n] \setminus \{i\}.$

\begin{definition}\label{def:potect-against-partition}
    A collection of payment functions 
    $\{t^1,\ldots,t^k\} \in \mathcal{T}_k$ 
    is said to protect action $i \in [n]$ against a partition 
    $\{S^1,\ldots, S^k\} \in \mathcal{S}_{k,i}$ 
    if, for each $j\in [k]$, $t^j$ protects action $i$ against $S^j$ .
\end{definition}

Note that Definition~\ref{def:potect-against-partition} does not yet imply that the ambiguous contract $\langle \{t^1,\ldots,t^k\} , i \rangle $ is IC. This is because different payment functions in the support of $\tau$ may entail different expected payments (and hence principal and agent utilities) for action $i$.

\subsection{A Balancing Routine that Ensures Consistency}\label{sec:balancing-routine}

As a first step, we show how to turn a collection of payment functions $\{t^1,\ldots,t^k\}$ that protects action $i$ against a partition $\{S^1,\ldots, S^k\} \in \mathcal{S}_{k,i}$ into an IC ambiguous contract $\langle \tau, i \rangle$ of support size $k$. Namely, we show how to turn it into a consistent collection of payment functions for which action $i$ is a best response.

To this end, consider Algorithm~\ref{alg:shifting}. The algorithm is given a target action $i$ and a collection of payment functions $\{t^1, \ldots, t^k\}$. The algorithm first computes the expected payment $T^\ell_i$ for action $i$ of each payment function $t^\ell \in \{t^1,\ldots,t^k\}$, 
it then sets $\theta:=\max\{T_i^1, \ldots, T_i^k\}$, and raises each component of the payment function $t^\ell$ by the same amount $c^\ell = \theta - T_i^\ell$.

\begin{algorithm}[t]
	\SetAlgoNoLine
	\KwIn{An action $i$, and a collection of payment functions $\{ 
t^1, t^2, \ldots t^k\}$}
	\KwOut{An ambiguous contract $\langle \tau, i \rangle$}
	{$\theta$ $\gets$ $\max\{ T_i^1, \ldots, T_i^k\}$ where, for each $j \in [k]$, $T_i^j = \sum_{\ell \in [m]} p_{i\ell}t^j_\ell$}\;
    $\tau = \emptyset$\;
	
		\For{$j \in [k]$
		}{
        $\bar{t}^j$ $\gets$ $t^j + (\theta - T^j_i) \cdot \vec{1}$  (for consistency)\;
        $\tau$ $\gets$ $\tau \cup \{\bar{t}^j\}$\;
				
		}
        
	\Return {$\langle \tau, i \rangle$}
	\caption{A Balancing Routine that Ensures Consistency}
	\label{alg:shifting}
\end{algorithm}

The following lemma 
shows that, when Algorithm~\ref{alg:shifting} is fed with a collection of payment functions $\{t^1, \ldots, t^k\}$ that protect action $i$ against a partition $\mathcal{S} = (S^1, \ldots, S^k)$ of the actions other than $i$, then it returns an (ambiguous) IC contract $\langle \tau, i\rangle.$ 
The proof appears in Appendix~\ref{app:sec_alg}.

\begin{lemma}
\label{lem:consistent}
Fix an action $i \in [n]$ and a collection of payment functions $\{t^1, \ldots, t^k\}$ that protects action $i$ against a partition $\{S^1, \ldots, S^k\} \in \mathcal{S}_{k,i}$.
Then Algorithm~\ref{alg:shifting} returns an IC ambiguous contract $\langle\tau, i\rangle$ with $\tau = \{\bar{t}^1, \ldots, \bar{t}^k\}$ satisfying 
$\UPmid{i} {\bar{t}^1}=\ldots=\UPmid{i}{\bar{t}^k}=\min_{j \in [k]}\{\UPmid{i}{t^j}\}$. 
\end{lemma}

\subsection{Finding IC Ambiguous Contracts that are Optimal for a Given Partition}\label{sec:optimal-for-partition}

As a next step, we show how to find an IC $k$-ambiguous contracts that is optimal for a given target action and fixed partition of the other actions, in the following sense.

\begin{definition}\label{def:optimal-wrt-partition}
    Given an action $i \in [n]$, an IC $k$-ambiguous contract $\langle \tau, i \rangle$ 
    is said to be optimal for action $i$ with respect to a partition $\mathcal{S} \in \mathcal{S}_{k,i}$ 
    if (i) $\tau$ protects action $i$ against $\mathcal{S}$, and (ii) for any other IC $k$-ambiguous contract $\langle \tau, i \rangle$ 
    such that $\hat{\tau}$ protects action $i$ against $\mathcal{S}$, it holds that $\UPmid{i}{\hat{\tau}} \le \UPmid{i}{\tau}$.
\end{definition}

Note that Definition~\ref{def:optimal-wrt-partition} only constitutes a necessary condition for $\langle \tau, i \rangle$ to be an optimal IC $k$-ambiguous contract. This is because, if $\langle \tau, i\rangle$ is an optimal IC $k$-ambiguous contract, then it is optimal for action $i$ with respect to \emph{some} partition $\mathcal{S} \in \mathcal{S}_{k,i}$.

Consider Algorithm~\ref{alg:AC_S}. The algorithm is given a target action $i$ and a partition $\mathcal{S} = \{S^1, \ldots, S^k\} \in \mathcal{S}_{k,i}$ of the actions other than $i$. It iterates over all subinstances $\mathcal{I}_{\{i\} \cup \{S^{\ell}\}}$ obtained by restricting the original instance to actions $\{i\} \cup S^\ell$, and for each it checks if action $i$ can be implemented with a classic contract, and, if it is, it uses  $\textsf{MIN-PAY-LP}(i)$ on subinstance $\mathcal{I}_{\{i\} \cup \{S^{\ell}\}}$ 
to find an optimal contract $\langle t^\ell, i \rangle$ that implements action $i$.  If this process succeeds for all sets in the partition, then the algorithm uses Algorithm~\ref{alg:shifting} to turn the resulting collection of payment functions $\{t^1, \ldots, t^k\}$ into an IC ambiguous contract $\langle \tau, i \rangle$. Otherwise, it returns $\textsf{Null}$.

\begin{algorithm}[t]
	\SetAlgoNoLine
	\KwIn{An action $i \in [n]$, a partition $\mathcal{S} = \{S^1, \ldots, S^k\}$ of $\nminusi$}
	\KwOut{An IC $k$-ambiguous contract $\langle \tau,i \rangle$  
that is optimal for action $i$ with respect to $\mathcal{S}$, or $\textsf{Null}$ if there is no $\tau \in \mathcal{T}_k$ that protects action $i$ against $\mathcal{S}$}
	    \tcp{Check if action $i$ is implementable with a classic contract in all subinstances}
		\eIf{$i$ \emph{is not implementable with respect to} $S^j \cup \{i\}$ \emph{for some $j \in [k]$}
			}{
				\Return \textsf{Null}
			}    
		  {
          For all $j \in [k]$:~~$t^j$ $\gets$ the solution to \textsf{MIN-PAY-LP}$(i)$ on subinstance $\mathcal{I}_{\{i\} \cup S^j}$ \;
          $\langle \tau, i \rangle$ $\gets$ Algorithm~\ref{alg:shifting} ($i$, $\{t^1, \ldots, t^k\}$)\;}
		
    \Return {$\langle \tau, i \rangle$
	}
    \caption{Optimal IC $k$-Ambiguous Contract for a Given Partition}
	\label{alg:AC_S}
\end{algorithm}

The following lemma shows that Algorithm~\ref{alg:AC_S} indeed finds an optimal IC $k$-ambiguous contract for action $i$ and a fixed partition $\mathcal{S} = \{S^1, \ldots, S^k\}$ of the actions other than $i$; and shows that the resulting ambiguous contract consists of shifted min-pay contracts. The proof is in Appendix~\ref{app:sec_alg}.

\begin{lemma}
\label{lem:alg1-opt}
    Given an action $i$ and a partition $\mathcal{S} = \{S^1, \ldots, S^k\} \in \mathcal{S}_{k,i}$, Algorithm~\ref{alg:AC_S} returns an optimal IC $k$-ambiguous contract for action $i$ with respect to $\mathcal{S}$, or $\textsf{Null}$ if there is no 
    $\tau \in \mathcal{T}_k$ that protects action $i$ against $\mathcal{S}$. Moreover, in the case where the algorithm returns an ambiguous contract $\langle \tau, i\rangle$, each payment function in the support of $\tau$ corresponds to a shifted min-pay contract. 
\end{lemma}

Note that the min-pay LP (see Figure~\ref{fig:minpaylp}) has the following (monotonicity) property: For any action $i$, if we add a constraint corresponding to any action $i'$, it either renders a feasible LP infeasible or (weakly) increases the optimal objective value.
Combining this monotonicity property with Lemma~\ref{lem:alg1-opt} directly implies the following observation. 

\begin{observation}\label{obs:monotonicity}
    Consider an action $i$ and a partition $\mathcal{S} = \{S^1, \ldots, S^k\} \in \mathcal{S}_{k,i}$ such that Algorithm~\ref{alg:AC_S} returns a non-null output $\langle \tau, i \rangle$. Consider any $j$ and the partition $\mathcal{S}' \in \mathcal{S}_{k+1,i}$ that results from splitting any $S^j$ into two sets $S^j_1$ and $S^j_2$. 
    Then Algorithm~\ref{alg:AC_S} applied to action $i$ and partition $\mathcal{S'}$ 
    returns a non-null output $\langle \tau', i \rangle$ such that $U_P(\langle \tau',i \rangle) \geq U_P(\langle \tau,i \rangle)$.
\end{observation}

\subsection{From Optimality for a Given Partition to Optimal $k$-Ambiguous Contracts}\label{sec:optimal-overall}

As a final step, we now show that, for any fixed target action $i$, we can find an optimal IC $k$-ambiguous contract for that action, or decide that none exists. This is achieved by Algorithm~\ref{alg:AC_i}, which applies Algorithm~\ref{alg:AC_S} to all possible partitions $\mathcal{S} \in \mathcal{S}_{k,i}$ of the actions other than $i$. For each partition $\mathcal{S}$, the algorithm thus finds an IC $k$-ambiguous contract $\langle \tau, i \rangle$ that is optimal with respect to $\mathcal{S}$ (or that it is impossible to protect action $i$ against this partition).  
It then return the IC $k$-ambiguous contract $\langle \tau^{\mathcal{S}}, i\rangle$ that minimizes the expected payment required to implement action $i$ (if one exists). 

\begin{algorithm}[t]
	\SetAlgoNoLine
	\KwIn{An action $i$}
	\KwOut{An optimal IC ambiguous contract $\tauactioni$ with $\tausizek$ for implementing action $i$, or $\textsf{Null}$ if the action is not implementable}
	$\mathcal{C} \gets$ $\emptyset$\; 
		\For{$\mathcal{S}\in \mathcal{S}_{k,i}$}{output $\gets$ Algorithm $\ref{alg:AC_S}$ ($i$, $\mathcal{S}$)\;
			\If{\emph{output} $\neq$ \emph{\textsf{Null}}
			}{
				$\langle \tau^\mathcal{S},i \rangle$ $\gets$ output\;
                $\mathcal{C} \gets \mathcal{C} \cup \{\langle\tau^{\mathcal{S}}, i\rangle \}$;
			}
		}
            \Return \textsf{Null} if $\mathcal{C} = \emptyset$, or $\langle \tau^\mathcal{S},i \rangle$ where $\mathcal{S} \in\argmin_{\mathcal{S}:\; \langle \tau^{\mathcal{S}},i\rangle \in \mathcal{C}} T_i(\tau^\mathcal{S})$ otherwise
	
	\caption{Optimal $k$-Ambiguous Contract for a Given Action}
	\label{alg:AC_i}
\end{algorithm}

\begin{lemma}\label{lem:opt-k-ambiguous-for-action-i}
Given an action $i$, Algorithm \ref{alg:AC_i} finds an optimal IC $k$-ambiguous contract $\tauactioni$, or decides that no such contract exists. In the case where action $i$ is implementable, the resulting $k$-ambiguous contract consists of shifted min-pay contracts.
\end{lemma}

\begin{proof}
    First, consider the case where action $i$ is implementable by a $k$-ambiguous contract. That is, there exists 
    an IC ambiguous contract of size $k$ that implements action $i$. Let $\langle \tauhatsizek,i\rangle$ be an arbitrary such contract. 
    It suffices to show that Algorithm~\ref{alg:AC_i} returns an IC $k$-ambiguous contract $\langle \tau^\mathcal{S},i \rangle$ such that $U_P(\langle \tau^\mathcal{S},i \rangle) \geq U_P(\langle \hat{\tau}, i \rangle)$. 
    Let $\hat{\mathcal{S}} = \{\hat{S}^1, \ldots, \hat{S}^k\}$ be a partition of $\nminusi$ defined as follows:
    \[
    \hat{S}^1 = \{i'\in \nminusi \mid \UAmid{i'}{\hat{t^1}} \le \UAmid{i}{\hat{t}^1}\}
    \]
    \[
    \hat{S}^2 = \{i'\in \left(\nminusi \right) \setminus \hat{S}^1 \mid \UAmid{i'}{\hat{t}^2} \le \UAmid{i}{\hat{t}^2}\}
    \]
    \vspace*{-0.5cm}
    \[
     \vdots
    \]
    \vspace*{-0.5cm}
    \[
    \hat{S}^k = \left\{ i'\in \left(\nminusi \right) \setminus \bigcup_{\ell \in [k-1]} \hat{S}^\ell \; \bigg| \; \UAmid{i'}{\hat{t}^k} \le \UAmid{i}{\hat{t}^k} \right\}
    \]

    By design, the collection of payment functions $\hat\tau$ protects action $i$ against partition $\hat{\mathcal{S}}$.
    Since $\hat{\mathcal{S}}$ is in $\mathcal{S}_{k,i}$, Algorithm~\ref{alg:AC_i} iterates over it. It follows that $\mathcal{C} \neq \emptyset$ (since it contains at least $\langle \tau^{\hat{\mathcal{S}}}, i \rangle$). 
    Moreover, the contract $\langle \tau^{\mathcal{S}}, i \rangle$ returned by Algorithm~\ref{alg:AC_i} has principal utility at least as high as $\langle \tau^{\hat{\mathcal{S}}},i \rangle$, which, in turn, has principal utility at least as high as $\langle \hat\tau , i \rangle$ by Lemma~\ref{lem:alg1-opt}.

    Next, consider the case where action $i$ is not implementable by an ambiguous contract of size $k$. In this case, no partition of $\nminusi$ to $k$ sets would result in an output that is different from $\textsf{Null}$ (when applying Algorithm \ref{alg:AC_S} to action $i$ and the given partition). Thus $\mathcal{C} = \emptyset$ and Algorithm \ref{alg:AC_i} returns $\textsf{Null}$.
\end{proof}

We are now ready to prove Theorem~\ref{thm:opt-k-ambiguous}.

\begin{proof}[Proof of Theorem \ref{thm:opt-k-ambiguous}]
Fix any $k$ and consider an arbitrary action $i^\star$ that can be implemented by a $k$-ambiguous contract. Apply Algorithm~\ref{alg:AC_i} and Lemma~\ref{lem:opt-k-ambiguous-for-action-i} to obtain an optimal IC $k$-ambiguous contract $\langle \tau, i^\star\rangle$ of the claimed form.
\end{proof}

\begin{remark}[Adjustments to obtain Theorem~\ref{thm:wlogSOP}] \label{rem:adjustments} 
Our proof can be adjusted to yield an alternative proof of Theorem~\ref{thm:wlogSOP} (from~\citet{DuettingFPS24}), specifically the claim that the optimal ambiguous contract in the unrestricted case is composed of at most $\min\{n-1,m\}$ single-outcome payment (SOP) contracts. This only requires employing a slightly different balancing routine, which only works in the unrestricted case (see Appendix~\ref{app:additive-vs-multiplicative}).

First suppose, that we are allowed to use $k = n-1$ payment functions. Then the optimal $k$-ambiguous contract for action $i^\star$ will be obtained by considering the partition of actions $[n] \setminus \{i^\star\}$ into $n-1$ singletons, each containing a single action $i' \in [n] \setminus \{i\}$ (by Observation~\ref{obs:monotonicity})
We can then use that, without loss of generality, the optimal classic contract for protecting action $i^\star$ against a single action $i'$ is an SOP contract that concentrates all payments on an outcome $j$ that maximizes the likelihood ratio $p_{i^\star j} / p_{i'j}$; see, e.g., \citet[Chapter 4.5.1]{LaffontM09} and \citet[][Proposition 5]{DuttingRT19}. 
Since for any such outcome $j$ it holds that $p_{i^\star j} \geq p_{i'j}$, rather than balancing payments by adding a fixed amount $c_{i'}$ to all outcomes, we can also adjust payments by multiplying all payments by some $\rho_{i'} \geq 1$. Note that this multiplicative adjustment 
preserves the SOP property.

We thus obtain that the optimal ambiguous contract consists of at most $n-1$ SOP contracts. This bound can be further strengthened to $\max\{n-1,m\}$ because if two SOP contracts pay for the same outcome, one of them is redundant.
\end{remark}

\subsection{Implementability under $k$-Ambiguous Contracts}

The structural properties revealed above lend themselves to a characterization of implementability by a $k$-ambiguous contract. This is cast in the following proposition.
The proof of this proposition resembles the proof of Lemma~\ref{lem:opt-k-ambiguous-for-action-i}, and is given in Appendix~\ref{app:sec_alg} for completeness.

\begin{proposition}\label{prop:implementable-with-k-ambiguous}
    An action $i$ in an instance $\mathcal{I}$ is implementable by a $k$-ambiguous contract if and only if there exists a partition $\mathcal{S}$ of $[n]\setminus\{i\}$ to $k$ sets, such that for every set $S\in \mathcal{S}$, 
    action $i$ is implementable by a classic contract in subinstance $\mathcal{I}_{\{i\} \cup S}$. That is, for each $S\in \mathcal{S}$ there does not exist a convex combination $\lambda_{i'} \in [0,1]$ of the actions $i' \in S$ that yields the same distribution over rewards $\sum_{i' \in S}\lambda_{i'} p_{i'j} = p_{ij}$ for all $j$ but at a strictly lower cost $\sum_{i'\in S} \lambda_{i'} c_{i'} < c_i$.
\end{proposition}

Proposition~\ref{prop:implementable-with-k-ambiguous} is trivially a generalization of Proposition~\ref{prop:implementable}, simply because Proposition~\ref{prop:implementable-with-k-ambiguous} instantiated with $k = 1$ coincides with Proposition~\ref{prop:implementable}. It is also a generalization of Proposition \ref{hadyn} since if the ambiguous contract has no size limit, the best-case for implementability is when each of the $n-1$ actions in $[n] \setminus \{i\}$ is in its own set (see Observation~\ref{obs:monotonicity}).

As a corollary of Proposition~\ref{prop:implementable-with-k-ambiguous} we obtain that, if an action is implementable by a $k$-ambiguous contract, then it is also implementable by any $k'$-ambiguous contract, for all $k' \geq k$. Hence, growing $k$ from $1$ to $n-1$ amounts to a gradual relaxation of the implementability constraints, with classic contracts and unrestricted ambiguous contracts representing the two extremes of this hierarchy.

\section{Ambiguous Contracts of Different Size: Succinctness Gap}
\label{sec:gap}
In this section, we quantify the loss in the principal's utility due to being restricted to ambiguous contracts with a small number of payment functions.
We define the $k$-succinctness gap of an instance as the ratio between the principal's utility in an optimal $k$-ambiguous contract and an optimal unrestricted ambiguous contract.
The $k$-succinctness gap of a family of instances is defined as the worst-case such ratio, over all instances in the family.
Formally, let $\mathcal A(\mathcal{I}_n)$ be the set of IC (unrestricted) ambiguous contracts, and $\mathcal A_k(\mathcal{I}_n)$ be the set of IC $k$-ambiguous contracts, for an instance $\mathcal{I}_n$ with $n$ actions.
Then, the $k$-succinctness gap is defined as follows.

\begin{definition}[$k$-succinctness gap]
Let $\Gamma_n$ denote the class of all instances with $n$ actions.
The $k$-succinctness gap $\rho_k(\mathcal{I}_n)$ of a given instance $\mathcal{I}_n \in \Gamma_n$ and the $k$-succinctness gap $\rho_k(\Gamma_n)$ of the  class of instances $\Gamma_n$, are 
\begin{align*}
      \rho_k(\mathcal{I}_n) = \frac{
\max_{\langle \tau,i\rangle \in\mathcal A_k(\mathcal{I}_n)}U_P(\langle \tau, i\rangle)
}{
\max_{\langle \tau,i\rangle\in\mathcal A(\mathcal{I}_n)}U_P(\langle \tau, i\rangle)
}
      \quad\quad \text{and} \quad\quad \rho_k(\Gamma_n) = \inf_{\mathcal{I} \in \Gamma_n} \rho_k(\mathcal{I}).
\end{align*}  
\end{definition}

Note that the succinctness gap $\rho_k$ of an instance/a class of instances 
is less than or equal to one. That is, the closer $\rho_k$ is to $1$ the smaller the gap. Moreover, lower bounds on $\rho_k$ constitute positive results, while upper bounds on $\rho_k$ constitute negative results. 

We know that $n-1$ contracts suffice to get the optimal ambiguous contract (see Theorem~\ref{thm:wlogSOP}). This means that, the $k$-succinctness gap $\rho_k(\Gamma_n)$ for every $k\geq n-1$ is $1$ (i.e., there is no gap). 
On the other extreme, where the ambiguous contract is composed of a single contract, the succinctness gap coincides with the notion of ambiguity gap studied in \cite{DuettingFPS24}, implying that the succinctness gap $\rho_1(\Gamma_n)$ is exactly $\frac{1}{n-1}$ (see Proposition 4 of their paper).

Below we establish both a lower bound on the succinctness gap (i.e., a positive result), as well as an upper bound (i.e., a negative result). As a key take-away we show that, even for the case of $k=n-2$, namely being just one contract shy of the sufficient number, the principal's utility might drop by a factor $2$; and this is tight. 

\paragraph{Positive Result: A Lower Bound on the Succinctness Gap.}

We first show that, for any $k\in[n-1]$, an ambiguous contract with support size $k$ can obtain at least a fraction $1/(n-k)$ of the optimal ambiguous contract.

\begin{theorem}[Lower bound]
\label{thm:k-lower}
    For any $k=1,\ldots,n-1$, the $k$-succinctness gap $\rho_k(\Gamma_n)$ satisfies
    $\rho_k(\Gamma_n) \geq \frac{1}{n-k}.$
\end{theorem}

Before presenting the proof of Theorem~\ref{thm:k-lower}, we give a sketch. For simplicity, we focus on the case where $k \geq 2$.
The proof builds on two ingredients: the optimal unrestricted ambiguous contract $\tau^\star$, which implements some action $a^\star$; and the optimal classic contract $t$ for the subinstance that consists of action $a^\star$ together with the $n-k$ lowest cost actions $a_{1}, \ldots, a_{n-k}$ different from $a^\star$. 
The proof then distinguishes two cases, depending on whether the agent's best response to $t$ is $a^\star$ or one of the actions in $\{a_1,\ldots,a_{n-k}\}$. 
If the best response is action $a^\star$, then we use the $k$-ambiguous contract consisting of payment function $t$ and the $k-1$ payment functions $t^{n-k+1}, \ldots, t^{n-1} \in \tau^\star$ that protect action $a^\star$ against the $k$ highest-cost actions $a_{n-k+1}, \ldots, a_{n-1}$ different from $a^\star$ (up to balancing). Otherwise, if the best response is some action $a_i$ for $i \in [n-k]$, then we use a $2$-ambiguous contract, composed of $t$ and a single contract $t'$ that pays a fixed amount, equal to the cost $c_i$ of action $a_i$, for all outcomes (again, up to balancing). 

Our proof that the restricted ambiguous contract constructed in either of the two cases achieves the desired principal's utility relies on Lemma~\ref{lem:consistent} and the following lemma, applied to the optimal classic contract $t$ for the subinstance used in the construction.

\begin{lemma}[\citet{DuettingFPS24}, Lemma 2] \label{lem:linear} 
Recall that, for a given action $i$, we use $\Welfare_i = \Reward_i - c_i$ to denote the action's welfare. Consider an instance $\mathcal{I}_n \in \Gamma_n$ which has a zero-cost action $i$ that yields an expected reward of $\Reward_i = 0$, then there exists a classic IC contract $\langle t, i' \rangle$ such that $U_P(\langle t, i' \rangle) \geq \frac{1}{n-1} \max_{i \in [n]} \Welfare_i$. 
\end{lemma}

We are now ready to present the proof of Theorem~\ref{thm:k-lower}.

\begin{proof}[Proof of Theorem~\ref{thm:k-lower}]
Consider an instance $\mathcal{I}_n$. 
    Let $\langle\tau^\star,a^\star\rangle$ be an optimal ambiguous contract for this instance, with principal's utility $U_P^{opt}$. 
    By Theorem~\ref{thm:wlogSOP},
    $\tau^\star$ is of size at most $n-1$. 
    Sort all other actions by cost, i.e., $c_1 \le c_2 \le \ldots \le c_{n-1}$.
    Then, by Lemma~\ref{lem:linear},  
    there exists a classic contract $\langle t, a'\rangle$ for the instance that includes only actions $a^\star$ and $a_1, \ldots ,a_{n-k}$ that achieves principal's utility at least $\frac{1}{n-k}\Welfare_{a^\star}$.
    For the case of $k=1$, the proof follows by taking the $1$-ambiguous contract $\tau = \{t\}$ and noting that the principal's utility achieved by this contract is $U_P(\langle t, a' \rangle) \geq \frac{1}{n-1}W_{a^\star} \geq \frac{1}{n-1}U_P^{opt}$.
    Thus, in the remainder of the proof we restrict attention to the case where $k \geq 2$. We distinguish between two cases, depending on whether the best response $a'$ to payment function $t$ is action $a^\star$ or an action from the set $\{a_1, \ldots ,a_{n-k}\}$:

    \medskip\noindent 
    {\bf Case 1:} The agent's best response to $t$ is action $a^\star$. 
    This means that $t$ protects action $a^\star$ against actions $a_1,\ldots,a_{n-k}$.
    Moreover, since $\langle \tau^\star, a^\star \rangle$ implements action $a^\star$ in the original instance, for all $i \in \{n-k+1,\ldots,n-1$\}, there is a payment function $t^i \in \tau^\star$ that protects action $a^\star$ against action $a_i$.
    
    The collection of payment functions $\{t, t^{n-k+1},\ldots, t^{n-1}\}$ is a set of at most $1+((n-1)-(n-k+1)+1)=k$ payment functions that contains a corresponding payment function that protects action $a^\star$ against each one of actions $a_1,\ldots,a_{n-1}$. The expected principal's utility for action $a^\star$ under the collection of payment functions $\{ t^{n-k+1},\ldots, t^{n-1} \}$ is $U_P^{opt}$, and the expected principal's utility for action $a^\star$ under payment function $t$ is at least $\frac{1}{n-k}\Welfare_{a^\star} \ge \frac{1}{n-k}U_P^{opt}$.

    Let $\langle \tau, a^\star\rangle$ be the 
    ambiguous contract returned by Algorithm \ref{alg:shifting} for action $a^\star$ and the collection of payment functions $\{t, t^{n-k+1},\ldots, t^{n-1}\}$. By Lemma~\ref{lem:consistent}, this is a IC ambiguous contract 
    with expected principal's utility of $\min\{U_P^{opt},\frac{1}{n-k} \Welfare_{a^\star}\} \geq \frac{1}{n-k}U_P^{opt}$, as required.
        
    \medskip\noindent
    {\bf Case 2:} The agent's best response to $t$ is action $a_i$ for some $i \in [n-k]$. 
    In this case, we know that $t$ protects action $a_i$ against $a^\star$ and $a_1,\ldots,a_{n-k}$. Moreover, since $t$ guarantees a principal utility of at least $\frac{1}{n-k}$ of the maximal welfare, and the utility in any contract cannot be higher than the  welfare, it holds that $\Welfare_{i} \geq \UPmid{i}{t}\ge \frac{1}{n-k} \Welfare_{a^\star}$. 
     
    Consider the payment function $t'$ which pays $c_i$ for every outcome.
    This payment function protects action $a_i$ against actions $a_{n-k+1},\ldots,a_{n-1}$ because the expected payment under $t'$ is $c_i$  for all actions, and the costs are sorted so that $c_i \le \min\{c_{n-k+1},\ldots ,c_{n-1}\}$. 
    The principal's utility for action $a_i$ under payment function $t'$ is $\UPmid{i}{t'}=R_i-c_i=\Welfare_i$. 
    
    Let $\langle \tau, a_i\rangle$ be the 
    ambiguous contract returned by Algorithm \ref{alg:shifting} for action $a_i$ and the collection of payment functions $\{t,t'\}$. This is an IC ambiguous contract 
    with expected principal's utility of $\min\{\UPmid{i}{t},\Welfare_i\} \geq \frac{1}{n-k} \Welfare_{a^\star} \geq \frac{1}{n-k} U_P^{opt}$, as needed. 
\end{proof}

\begin{remark}[Approximation algorithm]
    We note that the proof of Theorem~\ref{thm:k-lower} essentially provides a polynomial-time algorithm that finds an IC $k$-ambiguous contract that gives at least a $\frac{1}{n-k}$ fraction of the optimal principal utility in an unrestricted ambiguous contract. Clearly, this lower bound holds also with respect to the optimal IC $k$-ambiguous contract. Therefore, it also provides a   $\frac{1}{n-k}$-approximation for the optimal $k$-ambiguous contract problem.
\end{remark}

\paragraph{Negative Result: An Upper Bound on the Succinctness Gap.}
An especially interesting special case is an ambiguous contract of size  $n-2$ --- just one fewer than the number needed to achieve the optimal utility of an unrestricted ambiguous contract. 
A direct corollary of Theorem~\ref{thm:k-lower} is that such a contract guarantees at least half of the optimal principal's utility in the unrestricted case.
We show that this is tight. Namely, there exists an instance where no ambiguous contract with strictly less than $n-1$ payment functions yields more than half of the optimal principal's utility in the unrestricted case. 
This result is obtained as a special case of the following theorem, which shows that for any $n,k$ such that $n \geq 3$ and $1 \leq k \leq n-2$ there exists an instance where no ambiguous contract of size at most $k$ yields more than $\frac{1}{\lfloor\frac{n-2}{k}\rfloor + 1}$ of the optimal principal's utility in the unrestricted case.

\begin{theorem}[Upper bound]
\label{thm:kovern-new}
    For any $n,k$ such that $n \geq 3$ and $1 \leq k \leq n-2$ there exists an instance $\mathcal{I}_n$ such that the $k$-succinctness gap satisfies $\rho_{k}(\mathcal{I}_n) \leq \frac{1}{\lfloor\frac{n-2}{k}\rfloor + 1}$.
\end{theorem}

Notably, the tight ambiguity gap (namely, the ratio between classic contracts and unrestricted ambiguous contracts) established in \cite{DuttingFT24} is now obtained as a special case of our results. Indeed, for $k=1$, both the lower bound in Theorem~\ref{thm:k-lower} and the upper bound in Theorem~\ref{thm:kovern-new} give $\frac{1}{n-1}$.

\begin{proof}[Proof of Theorem~\ref{thm:kovern-new} (sketch)]
The proof of Theorem~\ref{thm:kovern-new}, deferred to Appendix~\ref{app:sec_gap}, considers a carefully constructed hard instance with a zero-cost action, $k$ groups of actions, and a target action. Within each group we sort actions from low to high cost. The construction ensures that, with a $k'$-ambiguous contract for any $k' \geq k+1$, it is possible to implement the target action ``at cost'' and thus extract the full welfare of that action as the principal's utility (which equals $\lfloor \frac{n-2}{k} \rfloor +1$). Thus, in order to establish the theorem, it suffices to show that, in this instance, any $k$-ambiguous contract can achieve a principal utility of at most $1$. 
Showing that any such contract that aims to implement any action other than the target action achieves a principal utility of at most $1$ is not difficult. Indeed, for the zero-cost action and the lowest-cost action from each group, we show that the welfare is bounded by $1$. For every other group action, we show that protecting it against the next lower action from the same group requires high payment.
It remains to bound the principal's utility from a contract that implements the target action.
The key insight here is that the highest-cost action from each group together with the zero-cost action give $k+1$ actions in total. Hence, from the pigeon's hole principle,  
any $k$-ambiguous contract must contain a payment function that protects against (at least) two of these actions. We use this to derive a lower bound on the payment, and consequentially an upper bound on the principal's utility.
\end{proof}

\section{Computational Complexity: Tractability and Hardness}
\label{sec:hardness}
Our characterization of optimal $k$-ambiguous contracts in Section~\ref{sec:alg} 
establishes a \emph{structural simplicity} of such contracts.
Our analysis 
also suggests a natural (but possibly na\"ive) algorithm (Algorithm~\ref{alg:AC_i}), which finds the best $k$-ambiguous contract for action $i$ by iterating over all possible partitions of the actions $[n] \setminus \{i\}$ into $k$ sets.
By applying this algorithm to all actions, we can find the optimal $k$-ambiguous contract.
The running time of this algorithm is proportional to the number of partitions of a set of size $n-1$ into $k$ non-empty sets, which 
is given by the Stirling number of the second kind $\stirling{n-1}{k}$.\footnote{For simplicity, in our analysis of Algorithm~\ref{alg:AC_i}, we allowed for partitions that contain empty sets, but it is without loss of generality to only consider partitions without empty sets (by Observation~\ref{obs:monotonicity}).} This number is polynomial if and only if $k=1$ or $k=n-O(1)$.

Thus, our approach for finding an optimal $k$-ambiguous contract, namely applying Algorithm~\ref{alg:AC_i} to each action $i$, recovers the results that optimal classic contracts (i.e., $k=1$) and optimal unrestricted ambiguous contracts (i.e., $k=n-1$) can be computed in polynomial time  \citep{GrossmanHart83,DuettingFPS24}.
Extending these results, it also yields an efficient algorithm for any $k=n-c$ for some constant $c$.
However, for all other regimes, the running time of our algorithm is 
super-polynomial.
This raises the natural question of whether a more efficient algorithm could bypass this complexity.

In this section, we present results that establish the intrinsic hardness of the problem.
In Section~\ref{sec:hardness-opt}, we show that computing an optimal $k$-ambiguous contract for instances with $n$ actions is as hard as the classic $(n,k)$-\textsc{Makespan Minimization} problem \cite[e.g.,][Section A5.2]{GareyJohnson79}.
This in particular implies that the problem is $\mathsf{NP}$-hard for 
any $k\geq 2$ such that $n-k=\omega(1)$ \citep{GareyJohnson79,GrahamLLK79},\footnote{It is well known that the problem is $\mathsf{NP}$-hard already for $k=2$ (via a reduction from \textsc{Partition}), and strongly $\mathsf{NP}$-hard for $2 \leq k \leq n/3$ (via a reduction from \textsc{$3$-PARTITION}) \citep{GareyJohnson79,GrahamLLK79}. By standard padding arguments, $\mathsf{NP}$-hardness for a fixed $k \geq 2$ extends to all regimes in which $n-k$ grows unboundedly. Consequently, the problem is $\mathsf{NP}$-hard whenever $k \geq 2$ and $n-k=\omega(1)$.}
which is precisely when our algorithm's running time is super-polynomial.
In Section~\ref{sec:hardness-approx}, we further establish an inapproximability result. Via a reduction from \textsc{3D-Matching} \cite[e.g.,][Section 3.1.2]{GareyJohnson79}, we show that for $k \approx n/3$ there is no polynomial-time algorithm that achieves a better than $3/4$-approximation for the optimal $k$-ambiguous contract problem (unless $\mathsf{P} = \mathsf{NP}$).

\begin{remark}[Structural implications]
Notably, the results in this section shed further light on our characterization result. In particular, our characterization shows that an optimal $k$-ambiguous contract induces a partition of the actions into groups, without prescribing how to form these groups. The hardness results in this section help explain this limitation: a more refined characterization that explicitly determines the grouping would require solving well-known combinatorial optimization problems, which are computationally intractable.
\end{remark}

\subsection{Hardness of Optimal Contracts}
\label{sec:hardness-opt}

We present a reduction from the well-known \textsc{Makespan Minimization} problem \cite[e.g.,][Section A5.2]{GareyJohnson79}
to the optimal succinct ambiguous contract problem.
This shows that the optimal $k$-ambiguous contract problem is $\textsf{NP}$-hard for any $k$ for which finding a partition of $n$ numbers into $k$ sets that minimizes the makespan is $\textsf{NP}$-hard, namely 
any $k \geq 2$ that doesn't satisfy $k=n-O(1)$. For example, it is $\mathsf{NP}$-hard for any constant $k \neq 1$ or any $k =\beta n$ for some constant $\beta \in (0,1)$. 

We start by defining the \textsc{Makespan Minimization} problem.

\begin{definition}[\textsc{Makespan Minimization}]
\label{def:makespan}
The $(n,k)$-\textsc{Makespan Minimization} problem is: Given a set of $n$ 
positive numbers $a_1, \ldots, a_n$, find a partition  of $a_1, \ldots, a_n$ into $k$ sets S$^1,\ldots,S^k$, such that the makespan 
$\max_{i \in [k]}\sum_{j \in S^i} a_j$ 
is minimized. 
\end{definition}

Our next theorem shows that the $k$-ambiguous contract problem in instances with $n$ actions is $\textsf{NP}$-hard whenever $(n,k)$-\textsc{Makespan Minimization} is $\textsf{NP}$-hard.

\begin{theorem}[Hardness of optimal contracts]
\label{thm:hardness-k}
The $k(n)$-ambiguous contract problem is \textsf{NP}-hard for all functions $k: [n] \rightarrow [n]$ such that $(n,k(n))$-\textsc{Makespan Minimization} is \textsf{NP}-hard. 
\end{theorem}

To prove Theorem \ref{thm:hardness-k}, we establish the following reduction:

\begin{proposition}
\label{prop:reduction}
There is a poly-time algorithm for mapping instances $\mathcal{J}_{n,k}$ of $(n,k)$-\textsc{Makespan Minimization} to instances $\mathcal{I}_{n+2}$ of the $k$-ambiguous contract problem, and a poly-time algorithm for turning an 
optimal solution to the $k$-ambiguous contract problem in $\mathcal{I}_{n+2}$ into an optimal solution to the $(n,k)$-\textsc{Makespan Minimization} in $\mathcal{J}_{n,k}$.
\end{proposition}

Given an instance $\mathcal{J}_{n,k}$ of $(n,k)$-\textsc{Makespan Minimization} with values $a_1, \ldots a_n$, we construct an instance $\mathcal{I}_{n+2} = \mathcal{I}(a_1, \ldots, a_n)$ of the $k$-ambiguous contract problem, 
with $n+2$ actions and $n+1$ outcomes, as follows (see Figure~\ref{fig:sop-tight1k}). 
Let $A = \sum_{i=1}^{n}a_i$. Action $0$ is a zero cost action that deterministically leads to a reward of $r_0 = 0$. Action $n+1$ leads to a uniform distribution over all outcomes in $[n]$ and has a cost of $c_{n+1} = \frac{1}{n} \max_{i \in [n]} a_i$. Action $i$ for $1 \leq i \leq n$ induces a uniform distribution over outcomes $\{0, \ldots, n\}\setminus \{i\}$, and has a cost of $c_i = c_{n+1} - \frac{1}{n} \cdot a_i$.

Next we show that, given an optimal solution to the $k$-ambiguous contract problem in $\mathcal{I}(a_1,\ldots,a_n)$, we can construct an optimal solution to $(n,k)$-\textsc{Makespan Minimization} in $\mathcal{J}_{n,k}$.

\begin{figure}[t]
    \centering
    \begin{tabular}{|c|c|c|c|c|c|c|c|c|c|}
    \hline
        \rule{0pt}{3ex} \hspace{2.0mm} rewards:  \hspace{2.0mm}&\hspace{2.0mm} $r_0 = 0$ \hspace{2.0mm}&\hspace{2.0mm} $r_1 = A$ \hspace{2.0mm}&\hspace{2.0mm} $r_2 = A$ \hspace{2.0mm}&\hspace{2.0mm} $\dots$ \hspace{2.0mm}&\hspace{2.0mm} $r_n = A$ \hspace{2.0mm}&\hspace{2.0mm} costs \hspace{2.0mm} \\[1ex] \hline
        \rule{0pt}{3ex}  action $0$: & $1$ & $0$ & $0$ & $\dots$ & $0$ & $c_0 = 0$  \\[1ex]
        \rule{0pt}{3ex}  action $1$: & $ \frac{1}{n}$ & $0$ & $\frac{1}{n}$ & $\dots$ & $\frac{1}{n}$ & $c_1 =c_{n+1}- \frac{1}{n} \cdot a_1$  \\[1ex] 
         \rule{0pt}{3ex}  action $2$:  & $ \frac{1}{n}$ & $\frac{1}{n}$ & $0$ & $\dots$ & $\frac{1}{n}$ & $c_2 = c_{n+1}- \frac{1}{n} \cdot a_2$  \\[1ex] 
         \rule{0pt}{3ex}  $\vdots$ & $\vdots$ & $\vdots$ & $\vdots$ & $\ddots$ &  $\vdots$ & $\vdots$  \\[1ex] 
         \rule{0pt}{3ex}  action $n$:  & $ \frac{1}{n}$  & $ \frac{1}{n}$  & $\frac{1}{n}$  &  $\dots$ & $0$ & $c_n = c_{n+1}- \frac{1}{n} \cdot a_n$  \\[1ex] 
         \rule{0pt}{3ex}  action $n+1$:  & $0$ & $\frac{1}{n}$ &  $\frac{1}{n}$ &  $\dots$ & $\frac{1}{n}$ & $c_{n+1} = \frac{1}{n}\max_{i \in [n]} a_i$  \\[1ex] \hline
    \end{tabular}
\caption{Instance $\mathcal{I}(a_1, \ldots, a_n)$ 
used in the proof of Theorem~\ref{thm:hardness-k}.
}
\label{fig:sop-tight1k}
\end{figure}

We start with the following observation.

\begin{observation}
\label{obs:condition1-k}
In instance $\mathcal{I}(a_1, \ldots, a_n)$, a 
payment function $t$ protects action $n+1$ against action $i \in [n]$ if and only if $t_i \ge t_0 +a_i$. 
\end{observation} 
\begin{proof}
Contract $t$ protects action $n+1$ against action $i$ if and only if:
\[
    \frac{1}{n} \sum_{j=1}^n t_j - c_{n+1} \ge \frac{1}{n}\cdot t_0 + \frac{1}{n} \sum_{j \neq i} t_j - (c_{n+1}-\frac{a_i}{n}).
\]
Rearranging, we get 
$t_i \ge t_0 + a_i$, as desired. 
\end{proof}

The next lemma establishes a direct correspondence between partitions of $[n]$ into $k$ sets and incentive-compatible $k$-ambiguous contracts for action $n+1$. In particular, for every partition $\bar{\mathcal{S}} = \{\bar{S}^1, \ldots, \bar{S}^k\}$ of $[n]$, we construct an IC $k$-ambiguous contract that implements action $n+1$, whose expected payment is exactly the makespan of the partition divided by $n$.

\begin{lemma}
\label{lem:partition-to-contract}
In the instance $\mathcal{I}(a_1, \ldots, a_n)$, for any partition $\bar{\mathcal{S}}=\{\bar{S}^1,\ldots,\bar{S}^k\}$ of $[n]$, there exists an IC $k$-ambiguous contract $\langle \bar{\tau}, n+1 \rangle$ that implements action $n+1$ such that the expected payment for action $n+1$ is exactly
\[
    \Payment{n+1}{\bar{\tau}} = \frac{1}{n}\max_{i\in[k]}\sum_{j\in\bar{S}^i} a_j.
\]
\end{lemma}

\begin{proof}
Let $\bar{\mathcal{S}}=\{\bar{S}^1,\ldots,\bar{S}^k\}$ be an arbitrary partition of $[n]$. Consider a collection of $k$ payment functions $\bar{t}^1, \ldots, \bar{t}^k$ defined as follows: for each $i \in [k]$ and outcome $j \in \{0, 1, \dots, n\}$,
\[
    \bar{t}^i_j = \begin{cases}
    a_j & \text{if } j \in \bar{S}^i, \\
    0 & \text{otherwise.}
    \end{cases}
\]
First, we calculate the expected payment of each payment function $\bar{t}^i$ for action $n+1$. Since action $n+1$ places a probability of $1/n$ on each outcome $j \in [n]$ and a probability of $0$ on outcome $0$, we have:
\[
    \Payment{n+1}{\bar{t}^i} = \frac{1}{n}\sum_{j \in \bar{S}^i} a_j.
\]

Next, we show that this collection of payment functions can be used to protect action $n+1$ against a partition of all non-target actions (which is the set $\{0, 1, \ldots, n\}$). 

By Observation~\ref{obs:condition1-k}, a payment function $t$ protects action $n+1$ against an action $j \in [n]$ if and only if $t_j \ge t_0 + a_j$. For any $i \in [k]$ and any action $j \in \bar{S}^i$, we have $\bar{t}^i_j = a_j$ and $\bar{t}^i_0 = 0$. Thus, the condition $\bar{t}^i_j \ge \bar{t}^i_0 + a_j$ holds with equality, meaning $\bar{t}^i$ protects action $n+1$ against all actions assigned to $\bar{S}^i$.

To ensure protection against action $0$, let $i^\star \in [k]$ be the index of a set containing 
a largest $a_j$;  
i.e., 
let $j^\star \in \argmax_{j\in[n]} a_j$, and let $i^\star$ be such that 
$j^\star \in \bar S^{i^\star}$.
The agent's expected utility for action $n+1$ under $\bar{t}^{i^\star}$ is:
\[
    \UAmid{n+1}{\bar{t}^{i^\star}} = \Payment{n+1}{\bar{t}^{i^\star}} - c_{n+1} = \frac{1}{n} \sum_{j \in \bar{S}^{i^\star}} a_j - \frac{1}{n} \max_{j \in [n]} a_j.
\]
Because the maximum element is contained in $\bar{S}^{i^\star}$ and all $a_j \ge 0$, the sum $\sum_{j \in \bar{S}^{i^\star}} a_j$ is at least $\max_{j \in [n]} a_j$, which implies $\UAmid{n+1}{\bar{t}^{i^\star}} \ge 0$. 
Furthermore, the agent's expected utility for action $0$ under $\bar{t}^{i^\star}$ is $\UAmid{0}{\bar{t}^{i^\star}} = \bar{t}^{i^\star}_0 - c_0 = 0$. We therefore have $\UAmid{n+1}{\bar{t}^{i^\star}} \ge \UAmid{0}{\bar{t}^{i^\star}}$, establishing that $\bar{t}^{i^\star}$ successfully protects action $n+1$ against action $0$.

Let us define a new partition $\hat{\mathcal{S}} = \{\hat{S}^1, \ldots, \hat{S}^k\}$ of the complete set of non-target actions $\{0, 1, \ldots, n\}$ by setting $\hat{S}^{i^\star} = \bar{S}^{i^\star} \cup \{0\}$ and $\hat{S}^i = \bar{S}^i$ for all $i \ne i^\star$. We have established that for each $i \in [k]$, the payment function $\bar{t}^i$ protects action $n+1$ against the set $\hat{S}^i$.

By Lemma~\ref{lem:consistent}, applying the balancing routine (Algorithm~\ref{alg:shifting}) to the collection $\{\bar{t}^1, \ldots, \bar{t}^k\}$ and action $n+1$ yields an IC $k$-ambiguous contract $\langle \bar{\tau}, n+1 \rangle$ that implements action $n+1$. 
The expected payment of the resulting consistent contract $\bar{\tau}$ for action $n+1$ is exactly the maximum expected payment among the input functions:
\[
    \Payment{n+1}{\bar{\tau}} = \max_{i \in [k]} \Payment{n+1}{\bar{t}^i} = \frac{1}{n} \max_{i \in [k]} \sum_{j \in \bar{S}^i} a_j.
\]
This concludes the proof.
\end{proof}
     
Our next proposition shows that for any (non-trivial) instance of the makespan minimization problem, the optimal IC $k$-ambiguous contract must implement action $n+1$.

\begin{proposition}
\label{prop:hardness-implmnt}
Suppose $n, k \geq 2$. Then any optimal IC $k$-ambiguous contract in the instance $\mathcal{I}(a_1, \ldots, a_n)$ must implement action $n+1$.
\end{proposition}

\begin{proof}
We show that the optimal IC $k$-ambiguous contract implements action $n+1$ by demonstrating that it achieves a principal's utility strictly greater than the maximum possible utility from implementing any other action.

We begin by computing the expected rewards of the actions in the instance 
$\mathcal{I}(a_1, \ldots, a_n)$. For action $n+1$, the expected reward is $R_{n+1} = \sum_{j=1}^n \frac{1}{n} r_j = \sum_{j=1}^n \frac{1}{n} A = A$. For any action $i \in [n]$, the expected reward is $R_i = \sum_{j \neq i} \frac{1}{n} r_j = \frac{n-1}{n} A$. For action $0$, the expected reward is $R_0 = 0$. 

Because limited liability requires all payments to be non-negative, the principal's expected utility from implementing any action is upper bounded by that action's expected reward. 
Therefore, for any contract implementing an action $i \in \{0, \ldots, n\}$, the principal's utility is at most $\frac{n-1}{n}A$.

We now show that the principal can achieve a utility strictly greater than $\frac{n-1}{n} A$ by implementing action $n+1$. 
First observe that, since $n, k \geq 2$ and all values $a_1, \ldots, a_n$ are positive, 
there exists a partition $\bar{\mathcal{S}} = \{\bar{S}^1, \ldots, \bar{S}^k\}$ of $[n]$ into $k$ sets such that the makespan $M = \max_{i \in [k]} \sum_{j \in \bar{S}^i} a_j$ is strictly smaller than the total sum $A = \sum_{j=1}^n a_j$. 

By Lemma~\ref{lem:partition-to-contract}, this partition $\bar{\mathcal{S}}$ induces an IC $k$-ambiguous contract $\langle \bar{\tau}, n+1 \rangle$ that implements action $n+1$ with an expected payment of exactly $M/n$. 
The principal's expected utility from this $k$-ambiguous contract is:
\[
    \UPmid{n+1}{\bar{\tau}} = R_{n+1} - \Payment{n+1}{\bar{\tau}} = A - \frac{M}{n}.
\]
Since $M < A$, it holds that:
\[
    \UPmid{n+1}{\bar{\tau}} > A - \frac{A}{n} = \frac{n-1}{n}A.
\]

Because this utility strictly exceeds the maximum possible principal's utility from implementing any other action, we conclude that any optimal IC $k$-ambiguous contract must implement action $n+1$.
\end{proof}

By Proposition~\ref{prop:hardness-implmnt}, whenever $n,k \geq 2$, 
any optimal $k$-ambiguous contract for $\mathcal{I}(a_1, \ldots, a_n)$ must implement action $n+1$.
Let $\langle\tau=\{t^1,\ldots,t^k\},n+1\rangle$ be such a contract.
The following proposition demonstrates how to derive from $\langle\tau,n+1\rangle$ an optimal solution to $(n,k)$-\textsc{Makespan Minimization} in instance $\mathcal{J}_{n,k}$ with values $a_1,\ldots,a_n$.

\begin{proposition}
\label{prop:reduction-makespan}
Suppose $n,k \geq 2$. 
Given an optimal IC $k$-ambiguous contract $\langle\tau=\{t^1,\ldots,t^k\},n+1\rangle$ for instance $\mathcal{I}(a_1, \ldots, a_n)$, the following partition  $\mathcal{S}=\{S^1,\ldots,S^k\}$ of $[n]$ gives an optimal solution to the $(n,k)$-\textsc{Makespan Minimization} problem in instance $\mathcal{J}_{n,k}$:
\[
    S^1 = \{i\in [n] \mid \UAmid{i}{t^1} \le \UAmid{n+1}{t^1}\}
\]
\[
    S^2 = \{i\in [n] \setminus S^1 \mid \UAmid{i}{t^2} \le \UAmid{n+1}{t^2}\}
\]
\vspace*{-0.5cm}
\[
     \vdots
\]
\vspace*{-0.5cm}
\[
    S^k = \left\{ i\in [n]\setminus \bigcup_{l \in [k-1]} S^l \bigg| \UAmid{i}{t^k} \le \UAmid{n+1}{t^k} \right\}
\]
\end{proposition}
    
\begin{proof}
We show that the makespan of the partition $\mathcal{S}$ derived from the optimal IC $k$-ambiguous contract $\langle \tau, n+1 \rangle$ is (weakly) smaller than that of any other partition.

We first show that the expected payment of $\tau$ for action $n+1$ is at least $\frac{1}{n}\max_{i\in[k]}\sum_{j\in S^i} a_j$.
Since partition $\mathcal{S}$ is constructed such that $t^i$ protects action $n+1$ against $S^i$, by Observation~\ref{obs:condition1-k} it holds that for all $i\in[k]$ and $j\in S^i$, $t^i_j\ge t^i_0 + a_j\ge a_j$, and so the expected payment of $t^i$ for action $n+1$ is at least 
\[
    \sum_{j\in S^i} \frac{1}{n}t^i_j \ge \frac{1}{n}\sum_{j\in S^i} a_j.
\] 
Since this is true for all $i\in [k]$ and $\langle\tau , n+1 \rangle$ is consistent, we get the desired inequality:
\begin{equation}
    \Payment{n+1}{\tau}\ge\frac{1}{n}\max_{i\in[k]}\sum_{j\in S^i} a_j.
    \label{eq:exp-payment}
\end{equation}

Next, consider an arbitrary partition $\bar{\mathcal{S}}=\{\bar{S}^1,\ldots,\bar{S}^k\}$ of $[n]$ into $k$ sets. By Lemma~\ref{lem:partition-to-contract}, there exists an IC $k$-ambiguous contract $\langle\bar{\tau},n+1\rangle$ that implements action $n+1$ such that the expected payment for action $n+1$ is exactly $\frac{1}{n}\max_{i\in[k]}\sum_{j\in\bar{S}^i}a_j$.

Since $\tau$ is an optimal IC $k$-ambiguous contract for $\mathcal{I}(a_1,\ldots,a_n)$, and in particular, optimal among contracts implementing action $n+1$, we have $\UPmid{n+1}{\tau} \ge \UPmid{n+1}{\bar{\tau}}$. 
Because the expected reward of action $n+1$ is fixed at $A$, this inequality implies that the expected payment under $\tau$ for action $n+1$ is at most the expected payment under $\bar{\tau}$. We therefore conclude that
\[
\frac{1}{n}\max_{i\in[k]}\sum_{j\in S^i}a_j
\le \Payment{n+1}{\tau}
\le \Payment{n+1}{\bar{\tau}} =\frac{1}{n}\max_{i\in[k]}\sum_{j\in\bar{S}^i}a_j,
\]
where the first inequality follows by Equation~\eqref{eq:exp-payment}. Since $\bar{\mathcal{S}}$ is an arbitrary partition of $[n]$, this shows that the partition $\mathcal{S}$ minimizes the makespan, and thus gives an optimal solution to the $(n,k)$-\textsc{Makespan Minimization} instance.
\end{proof}

We are now ready to prove Proposition~\ref{prop:reduction}.

\begin{proof}[Proof of Proposition~\ref{prop:reduction}]
The makespan minimization problem is trivial when either $n = 1$ or $k = 1$. So we may assume that $n,k \geq 2$. The poly-time reduction then follows from the observation that the construction of the instance $\mathcal{I}(a_1,\ldots,a_n)$ from the instance $\mathcal{J}_{n,k}$ 
is poly-time, and Proposition~\ref{prop:reduction-makespan} yields a poly-time algorithm for translating an optimal $k$-ambiguous contract in $\mathcal{I}_{n+2}$ into an optimal solution to $\mathcal{J}_{n,k}$.
\end{proof}

\subsection{Hardness of Approximation}
\label{sec:hardness-approx}

We now turn to the hardness of approximating the optimal $k$-ambiguous contract problem.
We begin by introducing 
the \Matching problem.

\begin{definition}
[\Matching]
The \Matching problem is: 
Given three disjoint sets $X$, $Y$, and $Z$, each of size $n$, and a set $T \subseteq X \times Y \times Z$ of triples, The goal is to determine whether there exists a subset $M \subseteq T$ such that every element of $X \cup Y \cup Z$ appears in exactly one triple of $M$. Such a subset $M$ is called a perfect 3-dimensional matching.
\end{definition} 
This problem is known to be $\textsf{NP}$-hard \cite[e.g.,][Section 3.1.2]{GareyJohnson79}. 
The following theorem establishes the hardness of approximating the optimal $k$-ambiguous contract problem.

\begin{theorem}[Hardness of approximation]
\label{thm:hardness-approx-util}
There exists a $k \in \Theta(n)$ 
such that no poly-time algorithm can approximate the optimal IC $k$-ambiguous contract to within a better factor than $3/4$,  
unless $\textsf{P} = \textsf{NP}$.
\end{theorem}

To prove Theorem \ref{thm:hardness-approx-util}, we establish the following reduction:

\begin{lemma} \label{lem:reduction_approx}
There exists a polynomial-time reduction that maps any instance 
$(X, Y, Z, T)$ of \Matching, 
with $|X| = |Y| = |Z| = n$, 
to an instance $\mathcal{I}_{3n+2}$ of the $k$-ambiguous contract problem, with $k=n$, 
such that:
\begin{enumerate}
    \item If the \Matching  instance admits a perfect matching, then $\mathcal{I}_{3n+2}$ admits an IC $k$-ambiguous contract with principal's utility of $\frac{4}{3}$.
    \item If the \Matching instance has no perfect matching, then any IC $k$-ambiguous contract obtains principal's utility of at most $1$.
\end{enumerate}
\end{lemma}

Consequently, approximating the optimal utility to within a factor greater than $\frac{3}{4}$ would allow distinguishing 
between satisfiable and unsatisfiable \Matching instances.

\begin{proof}[Proof of~Lemma~\ref{lem:reduction_approx}]
We start by describing the construction of the contract instance.
Given an instance $(X, Y, Z, T)$ of the \Matching problem with $|X| = |Y| = |Z| = n$ and $|T| = m$, we construct an instance $I_{3n+2}$ of the $k$-ambiguous contract problem with $k=n$, $3n + 2$ actions, and $m + 2$ outcomes (as described in Figure~\ref{fig:3dmatching}). 

We number the $m+2$ outcomes as $0, 1, \ldots, m+1$, and give them a reward of $r_j = 0$ for $j = 0, \ldots, m$ and $r_{m+1} = m+2$. We associate each triplet $T = \{T_1, \ldots, T_m\}$ with one of the outcomes $1, \ldots, m$ so that outcome $j$ corresponds to triplet $T_j$. 

One of the actions is action $0$, with a cost of $0$. This action leads to outcome $0$ with probability $1$ and to the remaining outcomes with probability zero.

Then for each element in $X \cup Y \cup Z$ we add one action, for a total of $3n$ actions. The actions corresponding to elements in $X$ are said to be in block $X$, and we denote them by $x_1, \ldots, x_n$; and similarly for elements in $Y$ and $Z$. We refer to actions in blocks $X,Y,Z$ as {\em block actions}. All of these actions have a cost of zero.
The probabilities of these actions 
for outcomes $1, \ldots, m$ corresponding to the triplets $T_1, \ldots, T_m$ are: 
\[
    p_{x_i,j} = \begin{cases}
    0 & \text{if $x_i \in T_j$}\\
    \frac{1}{m+2} & \text{otherwise}
    \end{cases}, \quad
    p_{y_i,j} = \begin{cases}
    0 & \text{if $y_i \in T_j$}\\
    \frac{1}{m+2} & \text{otherwise}
    \end{cases}, \quad
    p_{z_i,j} = \begin{cases}
    0 & \text{if $z_i \in T_j$}\\
    \frac{1}{m+2} & \text{otherwise}
    \end{cases}.
\]
We set the remaining probabilities of these actions to: 
\begin{align*}
    &p_{x_i,m+1} = \frac{1}{m+2}, \quad 
    p_{x_i,0} = 1-\sum_{j \geq 1} p_{x_i,j}, \\
    &p_{y_i,m+1} = \frac{1}{m+2}, \quad 
    p_{y_i,0} = 1-\sum_{j \geq 1} p_{y_i,j}, \\
    &p_{z_i,m+1} = \frac{1}{m+2}, \quad 
    p_{z_i,0} = 1-\sum_{j \geq 1} p_{z_i,j}.
\end{align*}

We finally add a \emph{target} action, with a cost of $\frac{2}{3}$. 
The target action leads to outcome $0$ with probability $0$, to each outcome $j=1, \ldots, m$ with probability $\frac{1}{m+2}$, and to outcome $m+1$ with probability $\frac{2}{m+2}$.

\begin{figure}[t]
\centering
\renewcommand{\arraystretch}{1.2}
\setlength{\tabcolsep}{10pt}
\begin{tabular}{|p{0.4cm} c|ccccc|c|}
\hline
 Block & Action & $r_0 = 0$ & $r_1 = 0$ & $\cdots$ & $r_m = 0$ & $r_{m+1} = m+2$ & Cost $c_i$ \\ 
\hline
 & 0 & 1 & 0 & $\cdots$ & 0 & 0 & 0 \\ 
\hline
 & $x_1$ & $1-\sum_{j = 1}^{m+1} p_{x_1,j}$ & $p_{x_1,1}$ & $\cdots$ & $p_{x_1,m}$ & $\tfrac{1}{m+2}$ & 0 \\
 \hspace*{0.2cm}$X$ & $\vdots$ & $\vdots$ & $\vdots$ & & $\vdots$ & $\vdots$ & $\vdots$ \\
 & $x_n$ &  $1-\sum_{j = 1}^{m+1} p_{x_n,j}$ & $p_{x_n,1}$ & $\cdots$ & $p_{x_n,m}$ & $\tfrac{1}{m+2}$ & 0 \\
\hline
 & $y_1$ & $1-\sum_{j = 1}^{m+1} p_{y_1,j}$ & $p_{y_1,1}$ & $\cdots$ & $p_{y_1,m}$ & $\tfrac{1}{m+2}$ & 0 \\
 \hspace*{0.2cm}$Y$ & $\vdots$ & $\vdots$ & $\vdots$ & $\vdots$ & $\vdots$ & & $\vdots$ \\
 & $y_n$ & $1-\sum_{j = 1}^{m+1} p_{y_n,j}$ & $p_{y_n,1}$ & $\cdots$ & $p_{y_n,m}$ & $\tfrac{1}{m+2}$ & 0 \\
\hline
 & $z_1$ & $1-\sum_{j = 1}^{m+1} p_{z_1,j}$ & $p_{z_1,1}$ & $\cdots$ & $p_{z_1,m}$ & $\tfrac{1}{m+2}$ & 0 \\
 \hspace*{0.2cm}$Z$ & $\vdots$ & $\vdots$ & $\vdots$ & $\vdots$ & $\vdots$ & & $\vdots$ \\
 & $z_n$ & $1-\sum_{j = 1}^{m+1} p_{z_n,j}$ & $p_{z_n,1}$ & $\cdots$ & $p_{z_n,m}$ & $\tfrac{1}{m+2}$ & 0 \\
\hline
 & target & 0 & $\tfrac{1}{m+2}$ & $\cdots$ & $\tfrac{1}{m+2}$ & $\tfrac{2}{m+2}$ & $\tfrac{2}{3}$ \\
\hline
\end{tabular}
\caption{Instance $I_{3n+2,\,m+2}$ used in the reduction. 
}
\label{fig:3dmatching}
\end{figure}

We next establish the gap claimed in the lemma.
Specifically, we show that 
if a perfect matching exists, then the target action can be implemented by a $k$-ambiguous contract, with $k = n$, with a principal's utility of $\frac{4}{3}$.  
Otherwise, no feasible IC $k$-ambiguous contract, with $k=n$, can yield a principal's utility greater than $1$. Note that since all actions apart from the target action have an expected reward of at most $1$, for the latter it is sufficient to bound the utility from ambiguous contracts that implement the target action.

\paragraph{Case 1: A perfect matching exists.}
Assume the \Matching instance admits a perfect matching $M = \{M_1, \dots, M_n\} \subseteq T$.
We show how to implement the target action with a consistent $n$-ambiguous contract that gives a principal's utility of $\frac{4}{3}$.

To this end, consider the $n$-ambiguous contract consisting of $n$ SOP contracts corresponding to the triplets in $M$. Each one of the SOP contracts pays $\frac{2(m+2)}{3}$ only to the outcome associated with its triplet (and zero otherwise).

First observe that the expected payment for each of these SOP contracts under the target action is $\frac{2}{3}$. This is because each of the SOP contracts pays an amount of $\frac{2(m+2)}{3}$ for exactly one of the outcomes $1, \ldots, m$ and the target action's probability for any of these outcomes is $\frac{1}{m+2}$.

This implies that the agent's utility from playing the target action is $0$, since payment equals cost. This also shows that the  contract protects the target action against action 0.

It remains to show that the contract protects the target action against all actions in blocks $X,Y,Z$.
Since $M$ is a perfect matching, it suffices to show that each SOP contract protects the target action against the three actions corresponding to its triplet. 
Recalling that the agent's utility from taking the target action is zero, it suffices to show that the agent's utility from taking any of these actions is also zero.
To see this, consider an SOP contract corresponding to a triplet $M_{\ell}$. Since the probability that each of the corresponding actions leads to the corresponding outcome is zero, the agent's utility for taking any of these actions is $0$. 

The argument for this case is completed, by noting that the principal's utility for the target action is $2-2/3 = 4/3$ (since the expected reward of the target action is $2$).

\paragraph{Case 2: No perfect matching exists.}
Assume the \Matching instance $(X, Y, Z, T)$ admits no perfect matching, and let $\tau$ denote an arbitrary IC $n$-ambiguous contract that implements the target action. We can, without loss of generality, assume that no $t \in \tau$ pays a positive amount for outcome $0$. Consequently, each $t \in \tau$ must protect the target action against a subset of the block actions, so that jointly all block actions are covered.

We distinguish between the following cases:

\medskip
\noindent\textbf{Case 2a:}
\emph{There exists a payment function $t \in \tau$ that protects the target action against two block actions from the same block.}

\medskip

Without loss of generality, let these actions be $x_1, x_2$.
Let $I_1$ be the set of indices $j \in [m]$ such that $x_1 \not\in T_j$. 
By the construction, this is the set of indices $j\in [m]$ s.t. $p_{x_1,j}=\frac{1}{m+2}$. Let $I_2$ be the analogous set of indices with respect to $x_2$, i.e., the set of indices $j\in [m]$ s.t. $p_{x_2,j}=\frac{1}{m+2}$.
Since every triplet $T_j \in T$ has exactly one element from $X$, for every $j\in[m]$, either $x_1 \not\in T_j$ or $x_2 \not\in T_j$, therefore it must hold that
\begin{equation}
\label{eq:union}
    I_1 \cup I_2 = [m].
\end{equation}

Since $t$ protects the target action against $x_1$, it holds that
\[
    \frac{1}{m+2} \sum_{i\in I_1} t_i + \frac{1}{m+2}t_{m+1} \le \frac{1}{m+2} \sum_{i=1}^{m} t_i + \frac{2}{m+2} t_{m+1}- \frac{2}{3} .
\]
Similarly, since $t$ protects the target action against $x_2$, it holds that
\[
    \frac{1}{m+2} \sum_{i\in I_2} t_i + \frac{1}{m+2}t_{m+1} \le \frac{1}{m+2} \sum_{i=1}^{m} t_i + \frac{2}{m+2}t_{m+1} - \frac{2}{3} .
\]

Summing these two inequalities and rearranging, we get
\begin{equation}
    \label{eq:sum}
    \frac{1}{m+2} \left( \sum_{i\in I_1} t_i + \sum_{i\in I_2} t_i \right) + \frac{4}{3}\le \frac{2}{m+2} \sum_{i=1}^{m} t_i + \frac{2}{m+2}t_{m+1}.
\end{equation}

Since $I_1 \cup I_2 = [m]$ (see Equation~\eqref{eq:union}), it holds that 
$\sum_{i \in I_1}t_i + \sum_{i \in I_2}t_i = \sum_{i=1}^{m}t_i + \sum_{i \in I_1 \cap I_2}t_i$.
Substituting the last equality in Equation~\eqref{eq:sum},
we get:
\[
    \frac{1}{m+2}  \sum_{i\in I_1 \cap I_2} t_i + \frac{4}{3} \le \frac{1}{m+2} \sum_{i=1}^{m} t_i + \frac{2}{m+2}t_{m+1}.
\]
The right-hand side of the last inequality is precisely the target action's expected payment, so this payment is at least $\frac{4}{3}$. Since the target action's expected reward is $2$, the principal's utility is at most $\frac{2}{3}$ (and thus at most $1$ as claimed).

\medskip
\noindent\textbf{Case 2b:}
\noindent
\emph{Every payment function $t \in \tau$ protects the target action against at most one action from each block.}
\medskip

Since no $t \in \tau$ protects the target action against more than one action from each block, and we have $k = n$ payment functions at our disposal in order to cover all $3n$ block actions, every $t \in \tau$ must protect the target action against exactly three block actions, one from each block.

Since the \Matching instance has no perfect matching, there must exist a payment function $t \in \tau$ that protects the target action against a triplet of actions $(x_a, y_b, z_c)$ that does not correspond to any triplet in $T$. This implies the following:

\begin{observation}
\label{obs:action_a}
    For any $j \in [m]$, there exists an action $d \in \{x_a, y_b, z_c\}$ such that $p_{d,j}=\frac{1}{m+2}$.
\end{observation}

Let $I_x = \{j \in [m] \mid p_{x_a, j} = \frac{1}{m+2}\}$, and define $I_y$ and $I_z$ similarly for actions $y_b$ and $z_c$. By Observation~\ref{obs:action_a}, every $j \in [m]$ belongs to at least one of these three sets, meaning $I_x \cup I_y \cup I_z = [m]$.

Since the payment function $t$ protects the target action against $x_a$, the agent's utility for the target action must be at least the agent's utility for $x_a$. Recalling that the target action places a probability of $\frac{1}{m+2}$ on all $j \in [m]$ and $\frac{2}{m+2}$ on outcome $m+1$, the expected payment for the target action under $t$ is:
\[
    T_{\text{target}}(t) = \frac{1}{m+2} \sum_{j=1}^m t_j + \frac{2}{m+2} t_{m+1}.
\]
For action $x_a$, the cost is $0$. It places a probability of $\frac{1}{m+2}$ on all $j \in I_x$, a probability of $\frac{1}{m+2}$ on outcome $m+1$, and some non-negative probability on outcome $0$. Because $t_0 = 0$ (as established at the beginning of Case 2), the agent's expected utility for $x_a$ is:
\[
    \UAmid{x_a}{t} = \frac{1}{m+2} \sum_{j \in I_x} t_j + \frac{1}{m+2} t_{m+1}.
\]
Recalling that the target action has a cost of $\frac{2}{3}$, 
the incentive compatibility constraint $\UAmid{\text{target}}{t} \ge \UAmid{x_a}{t}$ yields:
\[
    T_{\text{target}}(t) - \frac{2}{3} \ge \frac{1}{m+2} \sum_{j \in I_x} t_j + \frac{1}{m+2} t_{m+1}.
\]
Similarly, since $t$ protects the target action against $y_b$ and $z_c$, we obtain two analogous constraints:
\[
    T_{\text{target}}(t) - \frac{2}{3} \ge \frac{1}{m+2} \sum_{j \in I_y} t_j + \frac{1}{m+2} t_{m+1},
\]
and
\[
    T_{\text{target}}(t) - \frac{2}{3} \ge \frac{1}{m+2} \sum_{j \in I_z} t_j + \frac{1}{m+2} t_{m+1}.
\]

Summing these three inequalities yields:
\[
    3 \cdot T_{\text{target}}(t) - 2 \ge \frac{1}{m+2} \left( \sum_{j \in I_x} t_j + \sum_{j \in I_y} t_j + \sum_{j \in I_z} t_j \right) + \frac{3}{m+2} t_{m+1}.
\]
Because $I_x \cup I_y \cup I_z = [m]$ and all payments are non-negative ($t_j \ge 0$), the sum of the payments over these three sets must be at least the sum over all $j \in [m]$. Thus:
\[
    3 \cdot T_{\text{target}}(t) - 2 \ge \frac{1}{m+2} \sum_{j=1}^m t_j + \frac{3}{m+2} t_{m+1} = \left( \frac{1}{m+2} \sum_{j=1}^m t_j + \frac{2}{m+2} t_{m+1} \right) + \frac{1}{m+2} t_{m+1}.
\]
The term in parentheses is exactly $T_{\text{target}}(t)$, and thus the inequality simplifies to 
\[
    3 \cdot T_{\text{target}}(t) - 2 \ge T_{\text{target}}(t) + \frac{1}{m+2} t_{m+1}.
\]
Subtracting $T_{\text{target}}(t)$ from both sides gives:
\[
    2 \cdot T_{\text{target}}(t) \ge 2 + \frac{1}{m+2} t_{m+1}.
\]
Since payments are non-negative, this implies that  $T_{\text{target}}(t) \ge 1$. 

To summarize, 
the target action's expected payment is at least $1$, while its expected reward is $2$, implying that the principal's expected utility from implementing the target action under payment function $t$ is at most $2 - 1 = 1$. This concludes Case 2b, and completes the proof of Lemma~\ref{lem:reduction_approx}.
\end{proof}

We remark that the proof of Lemma~\ref{lem:reduction_approx} also shows that approximating the payment of an implementable action to a factor better than $3/2$ is \textsf{NP}-hard. To see this observe that in the case where a perfect matching exists, the target action can be implemented with an expected payment of $2/3$, while in the case where no such matching exists the expected payment required to implement the target action is at least $1$.

\section{Conclusion and Future Work}
\label{sec:conlcusion}
In this work, we study the structure and computational complexity of succinct ambiguous contracts. We show that optimal succinct ambiguous contracts admit a perhaps surprising separability property: Up to leveling, they are comprised of $k$ classic contracts, for a suitable partition of the actions.
We use this insight to study the succinctness gap, defined as the worst-case loss the principal may experience from using a $k$-ambiguous rather than an unrestricted ambiguous contract. 
Finally, we address the question of identifying the optimal partition and show that this problem is computationally intractable by establishing hardness results for computing optimal 
$k$-ambiguous contracts.

Our work opens several promising directions for future research. First, despite the apparent complexity of determining how actions should be grouped, it would be highly interesting, albeit likely challenging, to uncover additional structural insights into optimal partitions. Such insights could potentially lead to efficient approximation algorithms or to stronger hardness-of-approximation results.
More broadly, it would be exciting to explore whether similar separability characterizations arise in other ambiguous design settings, such as mechanism design or Bayesian persuasion.

\bibliographystyle{apalike}
\bibliography{bibliography}

\appendix

\section{Disclosure of AI and LLM Usage}
During the preparation of this work, we utilized 
Gemini DeepThink to 
proofread the text, verify formal proofs, and assist in extending the $\textsf{NP}$-hardness proof of the optimal $k$-ambiguous contract problem to its monotone counterpart.

\section{Proofs Omitted from Section~\ref{sec:model}}\label{app:model}

\begin{proof}[Proof of Lemma~\ref{lem:principal-utility}]
Consider 
an ambiguous contract $\ambcontract$ with $\tau = k$ such that $i$ is a best response to $\tau$.
Let the payment functions in $\tau$ be numbered so that 
\[
    \UPmid{i}{t^1}= \max_{t\in \tau} \UPmid{i}{t}.  
\]
Suppose 
\[
    \UPmid{i}{t^1} > \UPmid{i}{t^2}.
\]
Then it must be that 
\[
    \sum_{j=1}^mp_{ij}t^1_j<\sum_{j=1}^mp_{ij}t^2_j.
\]
Let $\theta\in[0,1)$ satisfy
\[
    \sum_{j=1}^mp_{ij} \theta t^2_j=\sum_{j=1}^mp_{ij}t^1_j,
\]
and consider the ambiguous contract $\ambcontractprime = \langle\{ t^1, \theta t^2,\ldots, t^k\},i\rangle$, constructed from $\ambcontract$ by replacing payment function $t^2$ with $\theta t^2$. Note that $|\tau'| = k$. For action $i$ and any action $i' \in [n]$ we have,
\begin{eqnarray*}
    \UAmid{i}{\theta t^2} &= &\UAmid{i}{t^1}~~~~~~~~<~~ \UAmid{i}{t^2} \qquad\text{and}\\
    \UAmid{i'}{\theta t^2} &=& \theta \Payment{i'}{t^2} - c_{i'} ~~~~\le~~\UAmid{i'}{t^2},
\end{eqnarray*}
which imply
\begin{eqnarray*}
    \min_{t\in \tau'}\UAmid{i}{t} &=&\min_{t\in \tau}\UAmid{i}{t} ~~~\quad\text{and}\\
    \min_{t\in \tau'}\UAmid{i'}{t} &\le&\min_{t\in \tau}\UAmid{i'}{t}~~~~~\forall i'\in[n],
\end{eqnarray*}
which establishes that action $i$ is a best response to $\tau'$.
Moreover, we have $\UPmid{i}{t^1} = \UPmid{i}{\theta t^2} > \UPmid{i}{t^2}$. 
Applying a similar argument to payment functions $t^3,\ldots, t^k$ yields the result.   
\end{proof}

\section{Proofs Omitted from Section~\ref{sec:alg}}\label{app:sec_alg}

\begin{proof}[Proof of Lemma~\ref{lem:consistent}]
We begin by showing that $\langle \tau, i\rangle$ is incentive compatible. We first show that $\langle \tau, i\rangle$ is consistent. To this end, observe that the expected payment for action $i$ under any payment function $\bar{t}^j \in \tau$ is
\begin{align}
    \bar{T}_i^j=\sum_{\ell \in [m]} p_{i\ell}\bar{t}^j_\ell =\sum_{\ell \in [m]} p_{i\ell}(t^j + (\theta - T^j_i))_\ell = \theta - T^j_i + \sum_{\ell \in [m]} p_{i\ell}t^j_\ell = \theta. \label{eq:cons}
\end{align}
    
Next we show that action $i$ is a best response to $\tau$. 
Fix an action $i'\ne i$, and let $j$ be the index such that $i' \in S^j$
We have
\begin{align}
    \UAmid{i'}{\bar{t}^j} &= \sum_{\ell \in [m]} p_{i'\ell}\bar{t}^j_\ell -c_{i'}= \sum_{\ell \in [m]} p_{i'\ell}(t^j_\ell + (\theta - T^j_i))-c_{i'} = T^j_{i'} -c_{i'} + \theta - T^j_i, \quad\text{and} \label{eq:util-A1}\\
    \UAmid{i}{\bar{t}^j} &= \bar{T}^j_i - c_i = \theta - c_i, \label{eq:util-A2}
\end{align}
where the last equality holds by \eqref{eq:cons}.
    
On the other hand, since $t^j$ protects action $i$ against the set of actions $S^j$, we have that
\begin{equation}
    T^j_i - c_i = \UAmid{i}{t^j} \geq \UAmid{i'}{t^j} = T^j_{i'} - c_{i'}.
    \label{eq:1111}
\end{equation}

It follows that
\[
    \UAmid{i}{\tau} - \UAmid{i'}{\tau} \geq \UAmid{i}{\bar{t}^j} - \UAmid{i'}{\bar{t}^j} = (T^j_i - c_i)  - (T^j_{i'} - c_{i'})\geq 0.
\]
For the first inequality, we use that 
$\UAmid{i}{\tau} = \UAmid{i}{\bar{t}^j}$, by the consistency of $\tau$, and that $\UAmid{i'}{\tau} \leq \UAmid{i'}{\bar{t}^j}$ since $\bar{t}^j\in \tau$. 
The equality holds by equations~\eqref{eq:util-A1} and \eqref{eq:util-A2}. The final inequality holds by \eqref{eq:1111}.

For the claim regarding the principal's utility note that for any $\bar{t}^j \in \tau$ it holds that
\[
    \UPmid{i}{\bar{t}^j} = R_i - \theta = R_i - \max\{ T_i^1, \ldots, T_i^k\} = \min_{j \in [k]} \{R_i - T_i^j\} = \min_{j \in [k]}\{\UPmid{i}{t^j}\},
\]
where the first equality follows by \eqref{eq:cons}.
\end{proof}

\begin{proof}[Proof of Lemma~\ref{lem:alg1-opt}] 
Consider first the case where Algorithm~\ref{alg:AC_S} returns an ambiguous contract $\langle \tau, i \rangle$. By construction, the collection of payment functions $\{t^1, \ldots, t^k\}$ determined by the algorithm in the for-loop protects action $i$ against partition $\mathcal{S}$. So, by Lemma~\ref{lem:consistent}, the ambiguous contract $\langle \tau, i \rangle$ returned by Algorithm~\ref{alg:AC_S} 
is a (consistent) $k$-ambiguous contract that implements action $i$. 
Moreover, with $\theta = \max_{j \in [k]} T_i^j$ defined as in Algorithm~\ref{alg:shifting}, we have $\tau = \{\bar{t}^1, \ldots, \bar{t}^k\}$, where for each $\bar{t}^j$ it holds that $\bar{t}^j = t^j + (\theta - T_i^j) \cdot \vec{1} = t^j + c^{j} \cdot \vec{1}$. So each $\bar{t}^j \in \tau$ corresponds to a shifted min-pay contract $\langle \bar{t}^j,i \rangle$ for the subinstance $\mathcal{I}_{\{i\} \cup S^j}$ obtained by restricting the original instance to actions $\{i\} \cup S^j$.
    
It remains to show that $\langle \tau,i \rangle$ is an \emph{optimal} IC $k$-ambiguous contract with respect to partition $\mathcal{S}$. To this end, consider an arbitrary IC $k$-ambiguous contract $\langle\tauhatsizek,i\rangle$ such that $\tauhatsizek$ protects action $i$ against partition $\mathcal{S}$. We show that $\UPmid{i}{\hat{\tau}} \leq \UPmid{i}{\tau}$. 
Let $j^\star \in \argmax_{j\in [k]} T_i^j$.
For this $j^\star$, it holds that $\theta = T_i^{j^\star}$, and thus
$\bar{t}^{j^\star} = t^{j^\star} + (\theta - T^{j^\star}_i) \cdot \vec{1} = t^{j^\star}$.  It follows that $\bar{t}^{j^\star} = t^{j^\star}$ is an optimal solution to \textsf{MIN-PAY-LP}$(i)$ on the subinstance $\mathcal{I}_{\{i\} \cup S^{j^\star}}$.  
Since $\hat{t}^{j^\star}$ is also a feasible solution to \textsf{MIN-PAY-LP}$(i)$ on the subinstance $\mathcal{I}_{\{i\} \cup S^{j^\star}}$,
it holds that
$
    \UPmid{i}{\hat{t}^{j^\star}} \le \UPmid{i}{t^{j^\star}}.
$
Combined with the fact that both $\tau$ and $\hat\tau$ are consistent, we get:
\[
    \UPmid{i}{\hat{\tau}} = 
    \UPmid{i}{\hat{t}^{j^\star}} \le \UPmid{i}{t^{j^\star}} = \UPmid{i}{\tau}.
\]

Consider next the case where Algorithm~\ref{alg:AC_S} returns $\textsf{Null}$.
The algorithm returns $\textsf{Null}$ only if there exists $j\in[k]$ for which running \textsf{MIN-PAY-LP}$(i)$ on the subinstance $\mathcal{I}_{\{i\} \cup S^j}$ returns $\textsf{Null}$.
But in this case, there is no payment function that protects action $i$ against $S^j$, concluding the argument.
\end{proof}

\begin{proof}[Proof of Proposition \ref{prop:implementable-with-k-ambiguous}]
$\Rightarrow$: Assume action $i$ is implementable by a $k$-ambiguous contract, and let 
$\langle \tausizek,i\rangle$ 
be a contract of size $k$ implementing action $i$. Define the following sets:
\[
    S^1 = \{i'\in \nminusi \mid \UAmid{i'}{t^1} \le \UAmid{i}{t^1}\}
\]
\[
    S^2 = \{i'\in \left(\nminusi \right) \setminus S^1 \mid \UAmid{i'}{t^2} \le \UAmid{i}{t^2}\}
\]
\vspace*{-0.5cm}
\[
    \vdots
\]
\vspace*{-0.5cm}
\[
    S^k = \left\{ i'\in \left(\nminusi \right)\setminus \bigcup_{l \in [k-1]} S^l \bigg| \UAmid{i'}{t^k} \le \UAmid{i}{t^k} \right\}
\]
and let $\mathcal{S}=\{S^1,S^2,\ldots,S^k\}$. By design, for all $j \in [k]$, it holds that $t^j$ protects action $i$ against $S^j$. 
Therefore, $t^j$ is a solution, not necessarily optimal, to \textsf{MIN-PAY-LP}$(i)$
in the subinstance $\mathcal{I}_{\{i\} \cup S^j}$, and in particular, this LP is feasible. From Proposition \ref{prop:implementable}, it holds that there does not exist a convex combination $\lambda_{i'} \in [0,1]$ of the actions $i' \in S^j$ that yields the same distribution over rewards $\sum_{i' \in S^j}\lambda_{i'} p_{i'j} = p_{ij}$ for all $j$ but at a strictly lower cost $\sum_{i'\in S^j} \lambda_{i'} c_{i'} < c_i$.

$\Leftarrow$: Assume 
there exists a partition $\mathcal{S}$ of $[n]\setminus\{i\}$ to $k$ sets, such that for every set $S\in \mathcal{S}$, 
action $i$ is implementable by a classic contract in $\mathcal{I}_{S\cup\{i\}}$. In this case, let $\{\hat{t}^1,\ldots, \hat{t}^k\}$ 
be a collection of payment function corresponding to solutions of \textsf{MIN-PAY-LP}$(i)$ in the subinstances $\mathcal{I}_{\{i\} \cup S^1},\ldots,\mathcal{I}_{\{i\} \cup S^k}$.
By running Algorithm \ref{alg:shifting} with the input $(i, \{\hat{t}^1,\ldots, \hat{t}^k\})$, we obtain an IC ambiguous contract $\langle \tauhatsizek, i \rangle$ 
of size $k$ implementing action $i$, as needed.
\end{proof}

\section{Additive vs.~Multiplicative Shifts}\label{app:additive-vs-multiplicative}

\begin{figure}[t]
    \centering
    \begin{tabular}{|c|c|c|c|c|c|c|c|c|c|}
    \hline
        \rule{0pt}{3ex} \hspace{2.0mm} rewards:  
        \hspace{2.0mm}&\hspace{2.0mm} $r_1 = 0$ \hspace{2.0mm}&\hspace{2.0mm} $r_2 = 33$ \hspace{2.0mm}&\hspace{2.0mm} $r_3 = 66$ \hspace{2.0mm}&\hspace{2.0mm} $r_4 = 124 $
        \hspace{2.0mm}&\hspace{2.0mm} costs \hspace{2.0mm} \\[1ex] \hline
        \rule{0pt}{3ex}  action $1$: & $1$ & $0$ & $0$ & $0$ & $c_1 = 0$  \\[1ex]
        \rule{0pt}{3ex}  action $2$: & $1/3$ & $5/12$ & $1/12$ & $1/6$ &  $c_2 = 0$  \\[1ex] 
        \rule{0pt}{3ex}  action $3$ & $5/12$ & $1/4$ & $1/6$ & $1/6$ & $c_3 = 0$   \\[1ex]
        \rule{0pt}{3ex}  action $4$ & $1/3$ & $1/3$ & $1/6$ & $1/6$ &  $c_4 = 2/3$  \\[1ex] 
        \rule{0pt}{3ex}  action $5$: & $0$ & $3/4$ & $1/4$ & $0$ & $c_5 = 4/3$  \\[1ex]
          \hline
    \end{tabular}
    \caption{Example demonstrating the necessity of additive rather than multiplicative shifts.}
    \label{fig:additive-vs-multiplicative}
\end{figure}

We present and analyze 
an example, which shows that for general $k$ we need to consider \emph{additive} shifts to the min-pay contracts rather than \emph{multiplicative} ones. Consider the instance with $n = 5$ actions and $m = 4$ outcomes in Figure~\ref{fig:additive-vs-multiplicative}. We want to identify an optimal $k$-ambiguous contract with $k = 2$ payment functions.

We claim that the optimal $2$-ambiguous contract in this example implements action $4$. To see this, first observe that the welfare of action $1$ is $\Welfare_1 = 0$, while the welfare of actions $2,3,5$ is $\Wel_2 = \Wel_3 = \Wel_5 = 479/12 \approx 39.9167$. So the principal's utility from any of these actions is at most $479/12 \approx 39.9167$.
We next derive the optimal $2$-ambiguous contract for action $4$. As established in Section~\ref{sec:alg}, we can find such a contract by considering all partitions of the actions $\{1,2,3,5\}$ into two non-empty sets, solving for the optimal classic contract in each, and balancing payments by applying an additive shift that equalizes expected payments:

\begin{itemize}
    \item Partitions $\{\{1,2,5\},\{3\}\}$ and $\{\{2,5\},\{1,3\}\}$.
    The best way to protect action $4$ against actions $\{1,2,5\}$ or $\{2,5\}$ is via $t^1 = (0,0,8,0)$ with an expected payment $T_4(t^1) = 8/6 = 4/3$, and the best way to protect action $4$ against action $\{3\}$ or actions $\{1,3\}$ is via $t^2 = (0,8,0,0)$ with an expected payment $T_4(t^2) = 8/3.$ Applying the additive balancing routine we obtain $\tau = \{t^1 + (4/3, 4/3, 4/3, 4/3),t^2\}$.
    This (consistent) ambiguous contract implements action $4$, with a principal's utility of $R_4 - 8/3 = 40.$
    \item Partitions $\{\{1,2,3\},\{5\}\}$ and $\{\{2,3\},\{1,5\}\}$ and $\{\{1\},\{2,3,5\}\}$. The best way to protect action $4$ against $\{2,3\}$ is via $t = (0,8,16,0)$ with an expected payment $T_4(t) = 16/3$. The expected payment for protecting action $4$ against $\{1,2,3\}$ and $\{2,3,5\}$ can only be higher. Since $16/3 > 8/3$ the principal's utility from implementing action $4$ via the partitions considered here is strictly less than $40$.
    \item Partitions $\{\{2\},\{1,3,5\}\}$ and $\{\{1,2\},\{3,5\}\}$. The best way to protect action $4$ against $\{3,5\}$ or $\{1,3,5\}$ is via $t = (0, 18, 0, 6)$ with an expected payment of $T_4(t) = 16/3$. Hence, the principal's utility from implementing action $4$ considered in this case is again strictly less than $40$. 
\end{itemize}

We conclude that the optimal ambiguous contract for implementing action $4$ is obtained via the partitions $\{\{1,2,5\},\{3\}\}$ or $\{\{2,5\},\{1,3\}\}$ and yields a principal's utility of $40$. 

Finally, observe that in order to balance the optimal classic contracts for the respective sets, $t^1 = (0,0,8,0)$ and $t^2 = (0,8,0,0)$ in a multiplicative rather than additive manner, we would have to scale $t^1$ to $(0,0,16,0)$.  However, the resulting ambiguous contract $\langle \tau',4\rangle$ with $\tau' = \{(0,0,16,0), (0,8,0,0)\}$ would no longer be IC. This is because under $\tau'$ action $5$ would yield the agent a strictly higher utility than action $4$. Namely,
\begin{align*}
    U_A(5 \mid \tau') &= \min\{U_A(5 \mid (0,0,16,0)),U_A(5 \mid (0,8,0,0))\} = \min\left\{\frac{1}{4} \cdot 16,\frac{3}{4}\cdot 8\right\} - \frac{4}{3} = \frac{8}{3}, \text{while}\\
    U_A(4 \mid \tau') &= \min\{U_A(4 \mid (0,0,16,0)),U_A(4\mid (0,8,0,0))\} = 8/3 - 2/3 = 2,
\end{align*}
and so $U_A(4 \mid \tau') < U_A(5 \mid \tau')$ in violation of IC.

\section{Proofs Omitted from Section~\ref{sec:gap}}\label{app:sec_gap}

\begin{figure}[t]
\begin{center}
\begin{tabular}{|c|cccccc|c|}
\hline
\text{rewards:}& $r_1 $ & $r_2 $ & $\ldots$ & $r_k $ & $r_{k+1} $ & $r_{k+2}$  &  \\[0ex] 
& $ 0$ & $0$ &  & $ 0$ & $0$ & $ \frac{1}{\gamma^{\Tilde{n}-2}}$  & \text{costs} \\[0ex] \hline
action $1$: &$\frac{1}{k}$&$\frac{1}{k}$&$\ldots$&$\frac{1}{k}$&$0$&$0$&$c_1=0$\\
action $i_1$: &$0$&$2\delta$&$\ldots$&$2\delta$& $1-\gamma^{\Tilde{n}-i}-(2k-2)\delta$ & $\gamma^{\Tilde{n}-i}$ & $c_{i_1} = \frac{1}{\gamma^{i-2}}-(i-1)+(i-2)\gamma$\\
$2 \leq i \leq \Tilde{n}-1$ &&&&&&&\\
action $i_2$: &$2\delta$&$0$&$\ldots$&$2\delta$& $1-\gamma^{\Tilde{n}-i}-(2k-2)\delta$  & $\gamma^{\Tilde{n}-i}$ & $c_{i_2} = \frac{1}{\gamma^{i-2}}-(i-1)+(i-2)\gamma$\\
$2 \leq i \leq \Tilde{n}-1$ &&&&&&&\\
$\vdots$&$\vdots$&$\vdots$&$\ddots$&$\vdots$&$\vdots$&$\vdots$&$\vdots$\\
action $i_k$: &$2\delta$&$2\delta$&$\ldots$&$0$& $1-\gamma^{\Tilde{n}-i}-(2k-2)\delta$  & $\gamma^{\Tilde{n}-i}$ & $c_{i_k} = \frac{1}{\gamma^{i-2}}-(i-1)+(i-2)\gamma$\\
$2 \leq i \leq \Tilde{n}-1$ &&&&&&&\\
action $\Tilde{n}$: &$\delta$&$\delta$&$\ldots$&$\delta$& $0$ & $1-k\delta$ & $c_{\Tilde{n}} = \frac{1}{\gamma^{\Tilde{n}-2}}-(\Tilde{n}-1)+(\Tilde{n}-2)\gamma$ 
\\[0ex] \hline
\end{tabular}
\end{center}

\caption{Instance $(c,r,p)$ used in the proof of Theorem~\ref{thm:kovern-new}.}
\label{fig:gap-general_k}
\end{figure}

\begin{proof}[Proof of Theorem \ref{thm:kovern-new}]
Consider any $n,k$ such that $n \geq 3$ and $1 \leq k \leq n-2$. Let $n' = \lfloor \frac{n-2}{k}\rfloor \cdot k +2$. Observe that $n' \leq \frac{n-2}{k}\cdot k +2 = n$. Let $\tilde{n} = \lfloor \frac{n-2}{k} \rfloor + 2$. Observe that $\tilde{n} \geq 3$ since $k \leq n-2$. 
Let $\gamma,\epsilon \in (0,\frac{1}{2})$ and let $\delta = \frac{1}{k} \cdot \epsilon \cdot \gamma^{\tilde{n}-2}$. Note that $\delta \in (0,\frac{1}{4k})$ since $\tilde{n} \geq 3$.

Below we will analyze the instance $\mathcal{I}_{n'}$ with $n' = \lfloor \frac{n-2}{k}\rfloor \cdot k +2 = (\tilde{n}-2) \cdot k + 2$ actions and $m = k+2$ outcomes given in Figure~\ref{fig:gap-general_k}, and we will show that $\rho_k(\mathcal{I}_{n'}) \leq \frac{k}{n' + k - 2}$. The claim then follows from considering the instance $\mathcal{I}_n$ which is identical to $\mathcal{I}_{n'}$ except that it has $n-n'$ dummy actions (e.g., copies of action 1), which brings the total number of actions to $n$. Since the dummy actions do not affect the arguments for the succinctness gap, we will be able to conclude that $\rho_k(\mathcal{I}_{n}) \leq \frac{k}{n' + k - 2}= \frac{1}{\lfloor \frac{n-2}{k} \rfloor + 1}$ as claimed.

\medskip\noindent
\textbf{The hard instance.} We first describe the instance $\mathcal{I}_{n'}$. The rewards of the $m = k+2$ outcomes are $r_j = 0$ for $j \in [k+1]$ and $r_{k+2} = \frac{1}{\gamma^{\tilde{n}-2}}$. 
Action $1$ is a zero-cost and zero-reward action. It puts a probability of $\frac{1}{k}$ on each of the first $k$ outcomes, and a probability of zero on outcomes $k+1$ and $k+2$. Then there are $k$ groups of actions $A_j = \{i_j \mid 2 \leq i \leq \Tilde{n}-1\}$ such that action $i_j$ has a cost of $c_{i_j} = \frac{1}{\gamma^{i-2}}-(i-1)+(i-2)\gamma$ (note that the cost only depends on $i$) and puts a probability of $2\delta$ on outcomes in $[k]\setminus \{j\}$, a probability of $1-\gamma^{\Tilde{n}-i}-(2k-2)\delta$ on outcome $k+1$, and a probability of $\gamma^{\Tilde{n}-i}$ on outcome $k+2$. 
The final action $\Tilde{n}$ will be our target action. That action has a cost of $c_{\Tilde{n}} = \frac{1}{\gamma^{\Tilde{n}-2}}-(\Tilde{n}-1)+(\Tilde{n}-2)\gamma$ and puts a probability $\delta$ on outcomes in $[k]$ and a probability of $1-k\delta$ on outcome $k+2$.

To see that this is a valid construction, first observe that for any $i \geq 2$ it holds that  $\frac{1}{\gamma^{i-2}} - (i-1) + (i-2)\gamma \geq 0$ (by the AM-GM inequality), implying that all costs $c_{i_j}$ for $2 \leq i \leq \tilde{n}-1$ and $1 \leq j \leq k$ and $c_{\tilde{n}}$ are non-negative.
It remains to verify that all quantities that we referred to as probabilities are well defined. To this end, first observe that the probabilities of all actions indeed sum up to $1$ as required. Additionally, we clearly have that $\frac{1}{k}, 2\delta, \delta \in [0,1]$ and that for all $i$ such that $2 \leq i \leq \tilde{n}-1$ it holds that $\gamma^{\tilde{n}-i} \in [0,1]$. We also have $1-k\delta \in [0,1]$ since $\delta \leq \frac{1}{k}.$ Finally, observe that for  all $i$ such that $2 \leq i \leq \tilde{n}-1$ it holds that $1-\gamma^{\Tilde{n}-i}-(2k-2)\delta \leq 1$. 
It remains to show that for any such $i$, it holds that $1-\gamma^{\Tilde{n}-i}-(2k-2)\delta \geq 0$. We split the analysis into two cases based on the value of $k$:

\begin{itemize}
    \item If $k=1$, the term $(2k-2)\delta$ evaluates to exactly $0$, and the condition simplifies to $1-\gamma^{\Tilde{n}-i} \geq 0$. Since $i \le \Tilde{n}-1$, the exponent satisfies $\Tilde{n}-i \ge 1$. Because we chose $\gamma \in (0, 1/2)$, it follows that $\gamma^{\Tilde{n}-i} \le \gamma < 1/2 < 1$. Thus, $1-\gamma^{\Tilde{n}-i} > 0$ holds trivially, and the probabilities are strictly positive without imposing any additional constraints on $\epsilon$.
    \item If $k \ge 2$, we can bound the expression by noting that $\gamma^{\Tilde{n}-i} \le \gamma$ (since $\Tilde{n}-i \ge 1$). Thus, to ensure the probability is non-negative, it suffices to show the stronger condition $1-\gamma \geq (2k-2)\delta$. Because $2k-2 > 0$ for $k \ge 2$, we can safely substitute $\delta = \frac{1}{k} \cdot \epsilon \cdot \gamma^{\Tilde{n}-2}$ and divide by $2k-2$. Rearranging the resulting inequality, we obtain that $\epsilon$ must satisfy $\epsilon \leq \frac{1-\gamma}{\gamma^{\Tilde{n}-2}} \cdot \frac{k}{2k-2}$. To see that this is indeed the case, observe that:
    \[
        \frac{1-\gamma}{\gamma^{\Tilde{n}-2}} \cdot \frac{k}{2k-2} \geq \frac{1-\gamma}{\gamma} \cdot \frac{k}{2k-2} \geq \frac{1}{2} > \epsilon,
    \]
    where the first inequality holds because the condition $1 \le k \le n-2$ guarantees $\Tilde{n} = \lfloor \frac{n-2}{k} \rfloor + 2 \ge 3$, meaning $\Tilde{n}-2 \ge 1$, and the function $\frac{1-\gamma}{\gamma^\kappa}$ is strictly increasing in $\kappa$; the second inequality holds because $\frac{1-\gamma}{\gamma} \ge 1$ (since $1-\gamma \ge \gamma$ for $\gamma \le 1/2$), and the function $\frac{k}{2k-2}$ is strictly decreasing for $k \ge 2$ and bounded below by its asymptote $1/2$ as $k \to \infty$; and the final strict inequality follows directly from our initial choice of $\epsilon \in (0, 1/2)$.
\end{itemize}

    
\medskip
\noindent \textbf{Analysis of the gap.} 
We will first show (in Claim 1) that with an ambiguous contract of size $k+1$, the principal can achieve a utility of $\Tilde{n}-1$ by implementing action $\Tilde{n}$. Then, in Claim 2, we will show that the principal's utility from any ambiguous contract of size at most $k$ that implements action $\Tilde{n}$ is at most $1$.
Finally, in Claim 3 we will show that the principal's utility under any ambiguous contract  that implements an action other than $\Tilde{n}$ is at most $1$. The bound on the succinctness gap follows since 
\[
    \frac{1}{\Tilde{n}-1} = \frac{k}{n'+k-2}.
\]

\medskip\noindent 
    
\textbf{Claim 1:} There exists an ambiguous contract of size $k+1$ that implements action $\tilde{n}$ and gives the principal a utility of $\tilde{n}-1$.

\medskip
To see that there exists an ambiguous contract of size $k+1$ that implements action $\tilde{n}$ and gives the principal a utility of $\tilde{n}-1$ consider the ambiguous contract $\langle \{t^1,t^2,\ldots,t^k,t^{k+1}\},\Tilde{n} \rangle$ 
with 
\begin{align*}
    &t^1 = (\frac{c_{\Tilde{n}}}{\delta},0,\ldots,0,0,0),\\
    &t^2 = (0,\frac{c_{\Tilde{n}}}{\delta},\ldots,0,0,0),\\
    &\ldots,\\
    &t^k =(0,0,\ldots,\frac{c_{\Tilde{n}}}{\delta},0,0 ), \qquad\quad\text{and}\\
    &t^{k+1} = (0,0,\ldots,0,0,\frac{c_{\Tilde{n}}}{1-k\delta}). 
\end{align*}

It is easy to verify that this contract is consistent, and entails an expected payment of $c_{\Tilde{n}}$ for action $\Tilde{n}$. It is IC, because for all actions $i \neq \Tilde{n}$ it gives a minimum payment of zero. 

The claim about the principal's utility follows, 
as the welfare from action $\tilde{n}$ is $W_{\Tilde{n}} = (\Tilde{n}-1) - (\Tilde{n}-2)\gamma - \epsilon$, and $W_{\Tilde{n}} \rightarrow \Tilde{n}-1$ as $\gamma, \epsilon \rightarrow 0$.
    
\medskip\noindent

\textbf{Claim 2:} 
The 
principal's utility from any ambiguous contract of size at most $k$ that implements action $\Tilde{n}$ is at most $1$.

\medskip
To prove this, we will show that the principal's utility for action $\Tilde{n}$ under any payment function $t$ that protects against two or more actions from the set of actions $\Gamma =\; \{ (\Tilde{n}-1)_1, (\Tilde{n}-1)_2, \ldots, (\Tilde{n}-1)_k, 1\}$ is at most $1$. This is sufficient because any ambiguous contract $\langle \tau, \tilde{n}\rangle$ of size $k$ has to have a payment function $t \in \tau$ that protects against two or more actions from $\Gamma$ (since $|\Gamma| = k+1$). Consider any such payment function $t$. 

\medskip

\noindent \textbf{Case 1.} \emph{Payment function $t$ protects action $\tilde{n}$ against action $1$ and action $(\tilde{n}-1)_i$ for some $i \in [k]$.}

\medskip

Since action $\Tilde{n}$ places zero probability on outcome $k+1$, we can without loss of generality assume $t_{k+1} = 0$. To obtain an upper bound on the principal's utility, we compare action $\Tilde{n}$ to action~$1$ and action $(\Tilde{n}-1)_i$. 
Since $t$ protects action $\Tilde{n}$ against action $1$, it follows that
\[
    \delta \cdot \sum_{i' \in[k]} t_{i'}  
    + (1-k\delta) \cdot t_{k+2} - c_{\Tilde{n}} \geq \frac{1}{k} \cdot 
    \sum_{i' \in[k]} t_{i'},
\]
or equivalently
\[
    t_{k+2} - \frac{1}{k} \cdot 
    \sum_{i' \in[k]} t_{i'}
    \geq \frac{c_{\Tilde{n}}}{1-k\delta}. 
\]
An important consequence of this is that 
\begin{equation}
    \label{eq:t4_t1_t2_k}
    t_{k+2} - \frac{1}{k} \cdot  
    \sum_{i' \in[k]} t_{i'}
    \geq 0. 
\end{equation}
On the other hand, since the agent does not want to deviate to action $(\Tilde{n}-1)_i$, it must hold that
\[
    \delta \cdot \sum_{i'\in[k]} t_{i'} 
    + (1-k\delta) \cdot t_{k+2} - c_{\Tilde{n}} 
    \geq  2 \delta \cdot \sum_{i' \in [k], i'\neq i} t_{i'} 
    + \gamma \cdot t_{k+2} - 
    c_{(\Tilde{n}-1)_i} \geq \gamma \cdot  t_{k+2} - 
    c_{(\Tilde{n}-1)_i}.
\]
It must therefore hold that
\begin{align*}
    \gamma \cdot t_{k+2} - 
    c_{(\Tilde{n}-1)_i} &\leq 
    \delta \cdot 
    \sum_{i'\in[k]} t_{i'}
    + (1-k\delta) t_{k+2} - c_{\Tilde{n}}\\
    &= t_{k+2} - c_{\Tilde{n}} -k\delta \cdot \left(t_{k+2}- \frac{1}{k} \cdot 
    \sum_{i'\in[k]} t_{i'}
    \right) \leq t_{k+2} - c_{\Tilde{n}},
\end{align*}
where the second inequality
holds by Equation~\eqref{eq:t4_t1_t2_k}.
Rearranging and substituting $c_{\Tilde{n}}$ and $c_{(\Tilde{n}-1)_i}$, this yields
\[
    t_{k+2} \geq \frac{1}{1-\gamma} \cdot (c_{\Tilde{n}}-c_{(\Tilde{n}-1)_i}) = \frac{1}{\gamma^{\Tilde{n}-2}} - 1.
\]
Hence, the principal's utility from action $\Tilde{n}$ under payment function $t$ 
is at most
\[
    (1-k\delta) \cdot (r_{k+2} - t_{k+2}) = (1-k\delta) \cdot \left(\frac{1}{\gamma^{\tilde{n}-2}} - t_{k+2}\right) \leq 1 - k\delta \leq 1,
\]
as claimed.

\medskip

\noindent\textbf{Case 2.} \emph{Payment function $t$ protects action $\tilde{n}$ against action $(\Tilde{n}-1)_i$ and action $(\Tilde{n}-1)_j$ for some $i,j \in [k]$ such that $i \neq j$.}

\medskip 
As in the previous case, we can assume without loss of generality that $t_{k+1} = 0$. 
Since $t$ protects action $\Tilde{n}$ against action $(\Tilde{n}-1)_i$, it holds that 
\[
    \delta \cdot 
    \sum_{i' \in [k]} t_{i'}
    + (1-k\delta) \cdot t_{k+2} - c_{\Tilde{n}} 
    \geq 2\delta \cdot 
    \sum_{i' \in [k], i' \neq i} t_{i'} + 
    \gamma \cdot t_{k+2} - c_{(\Tilde{n}-1)_i}.
\] 
It must therefore hold that
\begin{align*}
    \delta \cdot (t_i+t_j) 
    + (1-k\delta) \cdot t_{k+2} - c_{\Tilde{n}} 
    &\geq 2\delta \cdot t_j 
    + \delta \cdot \sum_{i' \in [k], i' \neq i,j} t_{i'}  
    + \gamma \cdot t_{k+2} - c_{(\Tilde{n}-1)_i} \\
    &\geq 2\delta \cdot t_{j} 
    + \gamma \cdot t_{k+2} - c_{(\Tilde{n}-1)_i},
\end{align*}
where we used that $\delta \geq 0$ and $t_{i'} \geq 0$ for all $i'$.
Similarly, since $t$ protects action $\Tilde{n}$ against action $(\Tilde{n}-1)_j$, it holds that
\[
    \delta \cdot (t_i+t_j)  
    + (1-k\delta) \cdot t_{k+2} - c_{\Tilde{n}} \geq  2\delta \cdot 
    t_i + \gamma \cdot t_{k+2}  - 
    c_{(\Tilde{n}-1)_j}
\]
By adding both inequalities, 
we obtain
\[
    2\delta \cdot (t_i+t_j) 
    + 2(1-k\delta) \cdot t_{k+2} - 2c_{\Tilde{n}} \geq 2\delta \cdot (t_i+t_j) 
    + 2\gamma \cdot t_{k+2}  - 2 \cdot c_{(\Tilde{n}-1)_i},
\]
where we used that $c_{(\Tilde{n}-1)_i} = c_{(\Tilde{n}-1)_j}$. 

Rearranging this inequality 
we obtain
\[
    \gamma \cdot t_{k+2}  - c_{(\Tilde{n}-1)_i} 
    \leq (1-k \delta) \cdot
    t_{k+2} - c_{\Tilde{n}} \leq t_{k+2} - c_{\Tilde{n}}.
\]
Solving for $t_{k+2}$ and substituting $c_{\Tilde{n}}$ and $c_{(\Tilde{n}-1)_i}$,  
we obtain
\[
    t_{k+2} \geq \frac{1}{1-\gamma} \cdot (c_{\Tilde{n}}-c_{(\Tilde{n}-1)_i}) = \frac{1}{\gamma^{\Tilde{n}-2}} - 1.
\]
Hence the principal's utility from action $\Tilde{n}$ 
is at most
\[
    (1-k\delta) \cdot (r_{k+2} - t_{k+2}) = (1-k\delta) \cdot \left(\frac{1}{\gamma^{\tilde{n}-2}} - t_{k+2}\right) \leq 1 - k\delta \leq 1,
\]
as claimed.

\medskip\noindent

\textbf{Claim 3.} For any ambiguous contract (of any size), that implements any action other than action $\Tilde{n}$, the principal's utility is upper bounded by $1$.

\medskip
Consider an IC ambiguous contract $\langle \tau, i^\star \rangle$ that implements action $i^\star \neq \Tilde{n}$. 
We distinguish between two cases.  
        
\medskip\noindent 
{\bf Case 1.} \emph{Contract $\langle \tau,i^\star \rangle$ implements action $1$, or an action of the form $2_j$ with $j \in [k]$.} 
    
\medskip 
These actions have welfare of at most $1$, implying that the principal's utility cannot exceed $1$. Indeed, for action $1$, the welfare is $\sum_{i' \in [k]} \frac{1}{k} \cdot r_{i'} - c_1 = 0$. On the other hand, for actions of the form $2_j$, the welfare is 
\[
    \sum_{i' \in [k], i' \neq j} 2\delta \cdot r_{i'} + (1-\gamma^{\tilde{n}-2}-(2k-2)\delta) \cdot r_{k+1} + \gamma^{\tilde{n}-2} \cdot r_{k+2} - c_{2_j} \leq \gamma^{\tilde{n}-2} \cdot r_{k+2} = 1.
\]

\medskip\noindent
{\bf Case 2.} \emph{Contract $\langle \tau, i^\star \rangle$ implements an action of type $i_j$ for some $i \ge 3$ and $j\in [k]$.} 
    
\medskip
Suppose $i^\star = i_j$.
Then there must be a payment function $t\in \tau$ that protects action $i_j$  
against action $(i-1)_j$.  
Since action $i_j$ 
places zero probability on  
outcome $j$,
we can assume without loss of generality that $t_j = 0$. 
According to this constraint, we must have 
\begin{align*}
    2\delta \cdot \sum_{i'\in[k], i' \neq j} t_{i'} 
    \;+\; & (1-\gamma^{\Tilde{n}-i}-(2k-2)\delta) \cdot t_{k+1} + \gamma^{\Tilde{n}-i} \cdot t_{k+2} - c_{i_j}  
    \geq \\
    &2\delta \cdot \sum_{i'\in[k], i' \neq j} t_{i'}
    + (1-\gamma^{\Tilde{n}-(i-1)}-(2k-2)\delta) \cdot t_{k+1} + \gamma^{\Tilde{n}-(i-1)} \cdot t_{k+2} - c_{(i-1)_j}, 
\end{align*}
or, equivalently, 
\[
    (1-\gamma^{\Tilde{n}-i}-(2k-2)\delta) \cdot t_{k+1} + \gamma^{\Tilde{n}-i} \cdot t_{k+2} - c_{i_j}  
    \geq  (1-\gamma^{\Tilde{n}-i+1}-(2k-2)\delta) \cdot t_{k+1} + \gamma^{\Tilde{n}-i+1} \cdot t_{k+2} 
    - c_{(i-1)_j}.
\]

Since $1-\gamma^{\Tilde{n}-i}-(2k-2)\delta < 1-\gamma^{\Tilde{n}-i+1}-(2k-2)\delta$ and $\gamma^{\Tilde{n}-i}>\gamma^{\Tilde{n}-i+1}$,  we obtain a lower bound on the expected payment $T_{i_j}(t)$ 
by setting $t_{k+1} = 0$ and finding the smallest $t_{k+2}$ such that 
\[
    \gamma^{\Tilde{n}-i} \cdot t_{k+2} - c_{i_j} 
    \geq \gamma^{\Tilde{n}-i+1} \cdot t_{k+2} -
    c_{(i-1)_j}.
\]
Rearranging and substituting 
$c_{i_j}$ and $c_{(i-1)_j}$, 
this yields 
\[
    t_{k+2} \geq \frac{1}{\gamma^{\Tilde{n}-i}} \cdot 
    \frac{c_{i_j} - c_{(i-1)_j}}{1-\gamma}
    = \frac{1}{\gamma^{\Tilde{n}-i}} \left(\frac{1}{\gamma^{i-2}} - 1 \right).
\]
Thus, the principal's utility from action $i_1$ is at most
\[
    \gamma^{\Tilde{n}-i} \cdot (r_{k+2} - t_{k+2}) = \gamma^{\Tilde{n}-i} \cdot \left(\frac{1}{\gamma^{\tilde{n}-2}} - t_{k+2}\right)\leq \frac{1}{\gamma^{i-2}} - \left(\frac{1}{\gamma^{i-2}} - 1 \right) = 1,
\]
which concludes the proof.
\end{proof}

\section{Monotone Succinct Ambiguous Contracts}
\label{app:monotone}

In this appendix, we explore a variation of our model, in which contracts are required to be monotone non-decreasing in rewards. 

\begin{definition}(monotone payment function)
    The payment function $t: [m] \rightarrow \mathbb{R}_{+}$ is monotone if outcomes entailing larger rewards engender 
    higher payments; formally: $r_j \geq r_{j'}$ $\Longrightarrow$ $t_{j} \geq t_{j'}$ for all $j,j' \in [m]$.
\end{definition}

Recall that by default we assume that outcomes are sorted so that rewards are monotone non-decreasing, i.e., such that $j \geq j'$ implies $r_j \geq r_{j'}$. An equivalent way to express monotonicity of payment functions is to require that $j \geq j'$ implies $t_j \geq t_{j'}$.

A monotone contract is then a tuple $\langle t, i \rangle$ consisting of a monotone payment function and a recommended action $i$. Similarly, a monotone $k$-ambiguous contract $\langle \tau, i \rangle$ is a collection of payment functions $\tau = \{t^1, \ldots, t^k\}$ together with a recommended action $i$ such that each $t^j \in \tau$ is a monotone payment function. The definitions of the agent's and the principal's utility, as well as the notion of an IC contract and an IC ambiguous contract remain unchanged.

Generalizing the LP formulation for classic contracts without the monotonicity requirement, the problem of finding the optimal monotone classic contract for a given action $i$ can be formulated as a linear program. See $\textsf{MON-MIN-PAY-LP}(i)$ in Figure~\ref{fig:minpaylp-monotone}.

\begin{figure}[ht]
\centering
\begin{subfigure}[b]{0.5\textwidth}
        \centering
        \begin{align*}
        \min \quad &\sum_{j} p_{ij} t_j \\
        \text{s.t.} \quad & \sum_j p_{ij} t_j -c_i \geq \sum_j p_{i'j} t_j -c_{i'}  &&\forall i' \neq i\\
        & t_j-t_{j-1} \geq 0 &&\forall j \geq 2\\
        & t_j \geq 0 &&\forall j
    \end{align*}
    \caption{\textsf{MON-MIN-PAY-LP}$(i)$}
    \end{subfigure}
\caption{The monotone minimum payment LP for action $i$.}
\label{fig:minpaylp-monotone}
\end{figure}

As in the case without monotonicity constraints, 
action $i$ can be implemented with monotone payments if $\textsf{MON-MIN-PAY-LP}(i)$ is feasible; and if it is, $\textsf{MON-MIN-PAY-LP}(i)$ finds an optimal monotone classic contract $\langle t, i \rangle$ for implementing action $i$. We refer to any such contract as a \emph{monotone min-pay contract} for action $i$.

\subsection{Optimal Monotone Contracts}

All results in Section \ref{sec:alg} extend to monotone $k$-ambiguous contracts in a natural way. This is because the balancing routine (Algorithm~\ref{alg:shifting}) adds the same constant amount to all payments, and thus preserves monotonicity. 

Generalizing the definition of a shifted min-pay contract, we define a \emph{shifted monotone min-pay contract} as a contract that is obtained by solving the ``monotone min-pay LP'' in Figure~\ref{fig:minpaylp-monotone} for some action $i^\star$ and subinstance $\mathcal{I}_{S^\ell \cup \{i^\star\}}$ for some subset $S^\ell \subseteq [n] \setminus \{i^\star\}$ of the actions other than $i^\star$, and then possibly adding a fixed amount $c_\ell \geq 0$ to each outcome $j \in [m]$. 

We then have the following characterization result, showing that the same separability that applies to possibly non-monotone succinct ambiguous contracts also applies when the payment functions are required to be monotone.

\begin{theorem}[Optimal monotone $k$-ambiguous contracts]\label{thm:opt-mon-k-ambiguous}
Fix any $k$. Suppose action $i^\star$ is implementable with a monotone $k$-ambiguous contract. Then there is an optimal IC monotone $k$-ambiguous contract $\langle \tau = \{t^1, \ldots, t^k\}, i^\star\rangle$ for implementing action $i^\star$ 
that takes the following form: 
\begin{itemize}[noitemsep]
\item There is a partition of the actions $[n]\setminus \{i^\star\}$ into $k$ sets $\{S^1, \ldots, S^k\}$. 
\item For each $\ell \in [k]$, contract $t^\ell$ is a shifted monotone min-pay contract for action $i^\star$ 
for the subinstance $\mathcal{I}_{\{i^\star\} \cup S^\ell}$, obtained by restricting the original instance to actions $\{i^\star\} \cup S^\ell$.
\end{itemize}
\end{theorem}

Similarly, we can extend the characterization of implementable actions from possibly non-monotone succinct ambiguous contracts to succinct ambiguous contracts that are required to be monotone.

\begin{proposition}\label{prop:implementable-with-mon-k-ambiguous}
An action $i$ in an instance $\mathcal{I}$ is implementable by a monotone $k$-ambiguous contract if and only if there exists a partition $\mathcal{S}$ of $[n]\setminus\{i\}$ to $k$ sets, such that for every set $S\in \mathcal{S}$, 
action $i$ is implementable by a classic monotone contract in subinstance $\mathcal{I}_{\{i\} \cup S}$.
\end{proposition}

\subsection{Bounds on the Succinctness Gap}

Clearly, we can also defined the $k$-succinctness gap with respect to monotone $k$-ambiguous contracts and the benchmark being the optimal unrestricted monotone ambiguous contract.

The non-trivial observation here is that the lower bound (positive result) on the $k$-succinctness gap for the case without monotonicity constraints established in Theorem~\ref{thm:k-lower}, in fact also applies to the case when we insist on monotonicity.

\begin{theorem}[Lower bound]
\label{thm:k-mon-lower}
For any $k=1,\ldots,n-1$, the monotone $k$-succinctness gap $\rho^{\textsf{mon}}_k(\Gamma_n)$ satisfies
$\rho^{\textsf{mon}}_k(\Gamma_n) \geq \frac{1}{n-k}.$
\end{theorem}

The reason is that Lemma~\ref{lem:linear} of \cite{DuettingFPS24} is in fact established through a linear contract, which is clearly
monotone. Indeed, all other payment functions that appear in the proof of Theorem~\ref{thm:k-lower} either come from the (already) monotone optimal unrestricted ambiguous contract or have a flat (constant) payment for all outcomes. 

\begin{remark}
The proof of the upper  
bound (i.e., negative result) of Theorem~\ref{thm:kovern-new} uses non-monotone payment functions, and so the argument does not extend in a straightforward way. 
\end{remark}

\subsection{Tractability and Hardness}

Following the same reasoning as in Section~\ref{sec:alg} we can interpret the respective algorithms (which use $\textsf{MON-MIN-PAY-LP}(i)$ rather than $\textsf{MIN-PAY-LP}(i)$) as a natural but possibly na\"ive algorithm. Since the addition of the monotonicity constraint does not change the tractability of the LP, the discussion of the running time in Section~\ref{sec:hardness} remains valid also when we insist on monotonicity.

Interestingly, also the hardness result linking  the computational complexity of the optimal $k$-ambiguous contract problem in instances with $n$ actions to the hardness of $(n,k)$-\textsc{Makespan Minimization} carries over to the monotone case.

\begin{theorem}[Hardness of optimal monotone contracts]
\label{thm:hardness-monotone-k}
    The $k(n)$-ambiguous contract problem with monotonicity constraints is \textsf{NP}-hard for all functions $k: [n] \rightarrow [n]$ such that $(n,k(n))$-\textsc{Makespan Minimization} is \textsf{NP}-hard. 
\end{theorem}

We establish the \textsf{NP}-hardness of the monotone $k$-ambiguous contract problem via a reduction from $(n,k)$-\textsc{Makespan Minimization}. 

\paragraph{Constructing the instance.} 
Given an instance $\mathcal{J}_{n,k}$ of $(n,k)$-\textsc{Makespan Minimization} with values $a_1, \ldots, a_n$, let $A = \sum_{i=1}^n a_i$.  We construct an instance $\mathcal{I}_{\textsf{mon}}(a_1, \ldots, a_n)$ of the monotone $k$-ambiguous contract problem with $n+2$ actions $\{0, 1, \dots, n+1\}$ and $n+1$ outcomes $\{0, 1, \dots, n\}$ as follows.
Let $R = A + 1$. The rewards are strictly increasing: $r_0 = 0$, and $r_j = j \cdot R$ for $j \in [n]$. The actions are defined as follows:
\begin{itemize}
    \item Action $n+1$ (the target action): Action $n+1$ places a probability of $0$ on outcome $0$ and a uniform probability of $1/n$ on outcomes $1, \ldots, n$. The cost of action $n+1$ is $c_{n+1} = \frac{1}{n} \max_{j \in [n]} \frac{a_j}{n-j+1}$.
    \item Action $i \in [n]$: Action $i \in [n]$ is  identical to action $n+1$, except that a probability mass of $1/n$ is shifted from outcome $i$ down to outcome $i-1$. That is, for $i = 1$, we have $p_{1,0} = 1/n$, $p_{1,1} = 0$, and $p_{1,j} = 1/n$ for $j \in \{2, \ldots, n\}$. For $i > 1$, we have $p_{i,0} = 0$, $p_{i,j} = 1/n$ for $j \in \{1, \ldots, i-2\}$, $p_{i,i-1} = 2/n$, $p_{i,i} = 0$, and $p_{i,j} = 1/n$ for $j \in \{i+1, \ldots, n\}$.
    The cost of action $i$ is $c_i = c_{n+1} - \frac{a_i}{n(n-i+1)}$. Note that $c_i \ge 0$ by the definition of $c_{n+1}$.
    \item Action $0$: Action $0$ deterministically yields outcome $0$ with probability $1$. The cost of action $0$ is $c_0 = 0$.
\end{itemize}

Note that the expected reward of action $n+1$ is $R_{n+1} = \frac{1}{n} \sum_{j=1}^n j \cdot R = \frac{n+1}{2} \cdot R$, while the expected reward of action $i \in [n]$ is $R_i = R_{n+1} - \frac{1}{n} r_i + \frac{1}{n} r_{i-1} = R_{n+1} - \frac{R}{n}$. The expected reward of action $0$ is $R_0 = 0$.

\paragraph{Properties of monotone contracts.} 
Before establishing the correspondence between the contracting problem and the original makespan minimization problem, we observe a few properties of monotone contracts in $\mathcal{I}_{\textsf{mon}}(a_1, \ldots, a_n)$.

Since the rewards are strictly increasing ($r_0 < r_1 < \dots < r_n$), the monotonicity constraint requires any valid payment function $t$ to be non-decreasing: $0 \le t_0 \le t_1 \le \dots \le t_n$. For $j \in [n]$ let $\Delta_j = t_j - t_{j-1} \ge 0$  denote the non-negative marginal payments. Thus, for any $j \in \{0\} \cup [n]$, the payment $t_j$ can be expressed as $t_j = t_0 + \sum_{\ell=1}^j \Delta_\ell$.
Using this, for any monotone payment function $t$, the expected payment for action $n+1$ is:
\begin{align}
    \Payment{n+1}{t} = \frac{1}{n} \sum_{j=1}^n t_j = \frac{1}{n} \sum_{j=1}^n \left( t_0 + \sum_{\ell=1}^j \Delta_\ell \right) = t_0 + \frac{1}{n} \sum_{i=1}^n (n-i+1) \Delta_i.
    \label{eq:mon-payment-formula}
\end{align}

The next observation provides a tool for deciding whether a monotone payment function $t$ protects action $n+1$ against action $i \in [n]$.

\begin{observation}
\label{obs:monotone-protect}
In instance $\mathcal{I}_{\textsf{mon}}(a_1, \ldots, a_n)$ monotone payment function $t$ protects action $n+1$ against action $i \in [n]$ if and only if $\Delta_i \geq \frac{a_i}{n-i+1}$.
\end{observation}

\begin{proof}
For any action $i \in [n]$, the payment function $t$ protects action $n+1$ against action $i$ if and only if $\UAmid{n+1}{t} \ge \UAmid{i}{t}$.
This condition simplifies to
\[
\Payment{n+1}{t} - c_{n+1} \ge \Payment{i}{t} - c_i,
\]
or equivalently,
\begin{equation}
\label{eq:protects}
\Payment{n+1}{t} - \Payment{i}{t} \ge c_{n+1} - c_i.
\end{equation}

By the structure of the instance,
\[
\Payment{i}{t} 
= \Payment{n+1}{t} - \frac{1}{n}t_i + \frac{1}{n}t_{i-1}
= \Payment{n+1}{t} - \frac{1}{n}\Delta_i .
\]

Substituting into Equation~(\ref{eq:protects}), we obtain that $t$ protects action $n+1$ against action $i$ if and only if
\[
\frac{1}{n}\Delta_i
\ge c_{n+1} - c_i
= \frac{a_i}{n(n-i+1)}.
\]
Rearranging the inequality yields the claim.
\end{proof}

\paragraph{Completing the reduction.} 
We start by showing that any partition $\bar{\mathcal{S}} = \{\bar{S}^1, \dots, \bar{S}^k\}$ of $[n]$ into $k$ sets,
implies a monotone IC $k$-ambiguous contract that implements action $n+1$ with an expected payment equal to the makespan achieved by $\bar{\mathcal{S}}$ divided by $n$.

\begin{lemma}\label{lem:mon-makespan-to-contract}
In the instance $\mathcal{I}_{\textsf{mon}}(a_1, \ldots, a_n)$, for any partition $\bar{\mathcal{S}}=\{\bar{S}^1,\ldots,\bar{S}^k\}$ of $[n]$, there exists a monotone IC $k$-ambiguous contract $\langle \bar{\tau}, n+1 \rangle$ that implements action $n+1$ such that the expected payment for action $n+1$ is exactly
\[
    \Payment{n+1}{\bar{\tau}} = \frac{1}{n}\max_{i\in[k]}\sum_{j\in\bar{S}^i} a_j.
\]
\end{lemma}

\begin{proof}
Given $\bar{\mathcal{S}} = \{\bar{S}^1, \ldots, \bar{S}^k\}$, we construct a collection of $k$ monotone payment functions $t^1, \dots, t^k$ as follows. For each $\ell \in [k]$, we set $t^\ell_0 = 0$ and define the marginals as $\Delta^\ell_i = \frac{a_i}{n-i+1}$ if $i \in \bar{S}^\ell$, and $0$ otherwise. Because $\Delta^\ell_i \ge 0$, each $t^\ell$ is a valid, monotonically non-decreasing function. By Observation~\ref{obs:monotone-protect}, $t^\ell$ protects action $n+1$ against all actions in $\bar{S}^\ell$. Furthermore, by Equation~\eqref{eq:mon-payment-formula}, the expected payment of $t^\ell$ satisfies: 
\[
    \Payment{n+1}{t^\ell} = \frac{1}{n} \sum_{i \in \bar{S}^\ell} (n-i+1) \frac{a_i}{n-i+1} = \frac{1}{n} \sum_{i \in \bar{S}^\ell} a_i.
\]

Applying the additive balancing routine (Algorithm~\ref{alg:shifting}) to this collection of monotone payment functions yields a monotone IC $k$-ambiguous contract $\bar{\tau}$ implementing action $n+1$ in the subinstance with actions $1, \ldots, n+1$.  
The expected payment of $\bar{\tau}$ is exactly 
\[
T_{n+1}(\bar{\tau}) =  \frac{1}{n}\max_{\ell \in [k]} \sum_{i \in \bar{S}^\ell} a_i = \frac{M_{\bar{\mathcal{S}}}}{n}.
\]

It remains to verify that $\bar{\tau}$ also protects action $n+1$ against action $0$. The agent's utility for action $n+1$ under $\bar{\tau}$ is $M_{\bar{\mathcal{S}}}/n - c_{n+1}$. Because $n-i+1 \ge 1$ for all $i \in [n]$, it holds that $a_i \ge \frac{a_i}{n-i+1}$. Since any valid makespan is at least the maximum single element, we have that $M_{\bar{\mathcal{S}}} \ge \max_i a_i \ge \max_i \frac{a_i}{n-i+1} = n \cdot c_{n+1}$. Thus, $M_{\bar{\mathcal{S}}}/n - c_{n+1} \ge 0$. Furthermore, the shift added to the bottleneck payment function $t^{\ell^\star}$ (i.e., $\ell^\star$ such that   $\Payment{n+1}{t^{\ell^\star}} = M_{\bar{\mathcal{S}}}/n$) is $0$, meaning its baseline payment remains $\bar{t}^{\ell^\star}_0 = 0$. Thus $\UAmid{0}{\bar{\tau}} = 0$, successfully protecting action $n+1$ against action $0$.
\end{proof}

Next we show that the optimal monotone $k$-ambiguous contract in instance $\mathcal{I}_{\textsf{mon}}(a_1, \ldots, a_n)$ must implement action $n+1$.

\begin{proposition}
\label{prop:monotone-must-implement-last-action}
In the instance $\mathcal{I}_{\textsf{mon}}(a_1, \ldots, a_n)$ 
the optimal monotone IC $k$-ambiguous contract must implement  
action $n+1$.  
\end{proposition}

\begin{proof}
We show that the principal's utility from the monotone $k$-ambiguous contract $\langle \tau, n+1\rangle$ derived from the optimal solution $\mathcal{S} = \{S^1, \ldots, S^k\}$ to the makespan minimization problem via Lemma~\ref{lem:mon-makespan-to-contract} is strictly higher than the utility from any other action. Recall that the principal's utility achieved by this contract is $U_P(\langle \tau, n+1\rangle) = R_{n+1} - \frac{M}{n}$, where $M = \max_{i \in [k]} \sum_{j \in S^i} a_j.$
We claim that $R>M$, which implies that $U_P(\langle \tau, n+1\rangle) > R_{n+1} - \frac{R}{n}$.
Indeed, this follows by the fact that $R = A + 1$, and $A \geq M$.

On the other hand, the expected reward of any action $i \in [n]$ is $R_i = R_{n+1} - \frac{1}{n} r_i + \frac{1}{n} r_{i-1} = R_{n+1} - \frac{R}{n}$, while action $0$ has an expected reward of $R_0 = 0$. Because the principal's utility cannot exceed the expected reward, implementing any action $i \in [n]$ or action $0$ yields a utility of at most $R_{n+1} - \frac{R}{n}$. 
Since the principal's utility from $\langle \tau, n+1 \rangle$ is strictly higher than that, we conclude that the optimal monotone IC $k$-ambiguous contract must implement action $n+1$. 
\end{proof}

By Proposition~\ref{prop:monotone-must-implement-last-action}, the optimal monotone IC $k$-ambiguous contract must implement action $n+1$.
We next show how to translate an optimal contract into an optimal solution to the makespan minimization problem.

\begin{proposition}
\label{prop:monotone-contract-to-partition}
Given an optimal monotone IC $k$-ambiguous contract $\langle\tau=\{t^1,\ldots,t^k\},n+1\rangle$ for instance $\mathcal{I}_{\textsf{mon}}(a_1, \ldots, a_n)$, the following partition  $\mathcal{S}=\{S^1,\ldots,S^k\}$ of $[n]$ gives an optimal solution to the $(n,k)$-\textsc{Makespan Minimization} problem in instance $\mathcal{J}_{n,k}$:
\[
    S^1 = \{i\in [n] \mid \UAmid{i}{t^1} \le \UAmid{n+1}{t^1}\}
\]
\[
    S^2 = \{i\in [n] \setminus S^1 \mid \UAmid{i}{t^2} \le \UAmid{n+1}{t^2}\}
\]
\vspace*{-0.5cm}
\[
     \vdots
\]
\vspace*{-0.5cm}
\[
    S^k = \left\{ i\in [n]\setminus \bigcup_{l \in [k-1]} S^l \bigg| \UAmid{i}{t^k} \le \UAmid{n+1}{t^k} \right\}
\]
\end{proposition}

\begin{proof}
We show that the makespan of the partition $\mathcal{S}$ derived from the optimal IC $k$-ambiguous contract $\langle \tau, n+1 \rangle$ is (weakly) smaller than that of any other partition.

We first show that the expected payment of $\tau$ for action $n+1$ is at least $\frac{1}{n}\max_{i\in[k]}\sum_{j\in S^i} a_j$.
Since partition $\mathcal{S}$ is constructed such that $t^i$ protects action $n+1$ against $S^i$, by Observation~\ref{obs:monotone-protect} it holds that for all $i\in[k]$ and $j\in S^i$, $\Delta^i_j \geq \frac{a_j}{n-j+1}$. So, by Equation~\eqref{eq:mon-payment-formula}, the expected payment of $t^i$ for action $n+1$ is at least 
\[
    T_{n+1}(t^i) \geq \frac{1}{n} \sum_{j=1}^{n} (n-j+1)\Delta^i_j \geq \frac{1}{n} \sum_{j \in S_i} (n-j+1)\frac{a_j}{n-j+1} = \frac{1}{n} \sum_{j \in S_i} a_j. 
\] 
Since this is true for all $i\in [k]$ and $\langle\tau , n+1 \rangle$ is consistent, we get the desired inequality:
\begin{align*}
    \Payment{n+1}{\tau}\ge\frac{1}{n}\max_{i\in[k]}\sum_{j\in S^i} a_j.
\end{align*}

Now consider an arbitrary partition $\bar{\mathcal{S}} = \{\bar{S}^1, \ldots, \bar{S}^k\}$ of $[n]$. Using Lemma~\ref{lem:mon-makespan-to-contract}, we obtain a monotone IC $k$-ambiguous contract $\langle \bar{\tau}, n+1 \rangle$ with an expected payment of
\[
T_{n+1}(\bar{\tau}) = \frac{1}{n} \max_{i \in [k]} \sum_{j \in \bar{S}^i} a_j.
\]
Since $\langle \tau, n+1 \rangle$ is an optimal $k$-ambiguous contract, we must have $T_{n+1}(\tau) \leq T_{n+1}(\bar{\tau})$ and thus
\[
\frac{1}{n} \max_{i \in [k]} \sum_{j \in S_i} a_j \leq T_{n+1}(\tau) \leq T_{n+1}(\bar{\tau}) = \frac{1}{n} \max_{i \in [k]} \sum_{j \in \bar{S}^i} a_j.
\]
Since this holds for an arbitrary partition $\bar{\mathcal{S}}$ it must hold for any partition $\bar{\mathcal{S}}$, and so $\mathcal{S}$ is indeed a partition that minimizes the makespan.
\end{proof}

We are now ready to prove Theorem~\ref{thm:hardness-monotone-k}.

\begin{proof}[Proof of Theorem~\ref{thm:hardness-monotone-k}]
The proof follows from noting that the construction of $\mathcal{I}_{\textsf{mon}}(a_1, \ldots, a_n)$ is poly-time, and Proposition~\ref{prop:monotone-contract-to-partition} yields a poly-time algorithm for translating an optimal monotone $k$-ambiguous contract into a makespan-minimizing partition.
\end{proof}

\end{document}